%% file: arxiv_version.tex
\documentclass[11pt]{article}

\input{arxiv_preamble}

\begin{document}
\noindent {\Large \bf  The combinatorial  code and the graph rules of Dale networks}\\
\noindent  { Nikola Mili\'cevi\'c and Vladimir Itskov}\\ 
Department of Mathematics, The Pennsylvania State University, University Park, PA 16802.\\

\input{manuscript_revision}

\printbibliography

\end{document}

%% file: arxiv_preamble.tex
\usepackage{graphicx,amssymb,amsmath,amsfonts,amsthm,color,url, placeins}
\usepackage{wrapfig}
\usepackage{helvet,graphicx,amssymb,amsmath,amsfonts,amsthm,color,url, placeins}
\usepackage{fullpage}
\usepackage{hyperref}
\usepackage{amssymb}
\usepackage{enumerate}
\usepackage{amsmath,xcolor}
\usepackage{amsthm}
\usepackage{amssymb}
\usepackage{float}
\usepackage{amsmath,mathrsfs}
\usepackage{adjustbox, bbm, color, epigraph, enumerate, enumitem, graphicx, mathabx, mathtools, tikz-cd, xcolor}
\usetikzlibrary{calc, cd, decorations.markings, decorations.pathmorphing, decorations.pathreplacing, fit, positioning, shapes, shapes.geometric, tikzmark}
\usepackage[capitalize,noabbrev]{cleveref}
\usepackage[backend=bibtex,maxbibnames=99,sorting=anyt]{biblatex}
\renewbibmacro{in:}{}
\bibliography{References}

\newtheorem{theorem}{Theorem}[section]

\newtheorem{corollary}[theorem]{Corollary}

\newtheorem{proposition}[theorem]{Proposition}

\newtheorem{lemma}[theorem]{Lemma}
\theoremstyle{definition}
\newtheorem{definition}[theorem]{Definition}

\theoremstyle{remark}
\newtheorem{remark}[theorem]{Remark}
\theoremstyle{plain}
\newtheorem*{assumption}{Ground Assumption}

\newcommand{\od}{\stackrel{\mbox {\tiny {def}}}{=}}
\newcommand{\code}[2]{\textnormal{code}(#1,#2)}
\newcommand{\exccode}[1]{\mathcal{C}(#1)}
\newcommand{\stbcode}[1]{\mathcal{SC}(#1)}
\newcommand{\inhib}{\mathcal{I}}
\newcommand{\exc}{\mathcal{E}}
\newcommand{\excinhib}{\mathcal{E_I}}
\newcommand{\excuninhib}{\mathcal{E_U}}
\newcommand{\excsup}[1]{\textnormal{supp}_+ \, #1}
\definecolor{Red}{rgb}{1.0, 0.0, 0.0}
\definecolor{NavyBlue}{rgb}{0.0, 0.0, 0.5}
\definecolor{OliveGreen}{rgb}{0.33, 0.42, 0.18}
\newcommand\new[1]{#1}
\newcommand{\supp}{{\text{supp}}}
\newcommand{\C}{{\mathcal{C}}}
\newcommand\Dale[1]{\mathbb D_{{#1}}}
\newcommand{\outnbhd}[2]{N^{+}_{#1}(#2)}

\makeatletter

%% file: manuscript_revision.tex
\begin{abstract}
 We describe the combinatorics of equilibria  and  steady states of  neurons in threshold-linear networks that satisfy  \new{Dale’s law}. The combinatorial code of a Dale network is characterized in terms of  two conditions: (i) a condition on the network connectivity graph, and (ii) a spectral condition on the synaptic matrix.  We find that in the weak coupling regime the combinatorial code depends only on the connectivity graph, and not on the particulars of the synaptic strengths. Moreover, we prove that the combinatorial code of a weakly coupled network is a sublattice, and  we provide a learning rule for encoding a sublattice in a weakly coupled  excitatory network.    In the strong coupling regime we prove  that the combinatorial code of  a generic Dale network  is intersection-complete and is therefore a convex code, as is common in some sensory systems in the brain.
\end{abstract}

\begin{keywords}
Dale's law, recurrent neural networks,
combinatorial neural codes
\end{keywords}

\begin{MSCcodes}
92B20
\end{MSCcodes}


\section{Introduction}
We study the  equilibria  and  steady states of threshold-linear networks that satisfy  \new{Dale's} law. 
\new{Dale's law, introduced by Sir Henry Dale in 1935 \cite{dale1935pharmacology,osborne2013dale},  postulates  that  a single neuron  utilizes the same set of chemical messengers. This usually implies that each neuron's efferent (outgoing) synapses are either all excitatory or all inhibitory. While this restriction on the signs  of the synaptic weights has a number of exceptions in {\it some} neural systems, it is observed in most neural circuits on the brain \cite{kandel:neural}.}

Threshold-linear networks  ha\new{ve} been  extensively studied in  \cite{HahnSeungSlotine,XieHahnSeung, Hahn2000, curto2012flexible, net-encoding, pattern-completion, nozari2020hierarchical}, and a number of results regarding the stable fixed points of these networks were obtained, especially in the case of a symmetric synaptic matrix. 
Networks that satisfy \new{Dale's} law have been previously investigated in the context of large random networks \cite{rajan2006eigenvalue,ahmadian2015properties,ipsen2020consequences} (this list is very incomplete), where statistical properties such as the spectrum,  the number of equilibria, and phase transitions have been investigated.  In contrast, we are interested in describing the exact combinatorics of the neural code of these networks, as opposed to investigating statistical features. 

To this end we developed a method for understanding the combinatorial codes of threshold-linear networks that obey \new{Dale's law}. A combinatorial code is the collection of patterns of neuronal activation at the equilibria that is possible \new{in} a given network. In particular, this discards the details of the firing rates, and only keeps track of what neurons are co-active \new{at a fixed point}.  It turns out \new{it is} possible  to directly translate features of the connectivity graph into the combinatorial code of a Dale network. We show that the connectivity features completely determine the combinatorial neural code in the weak coupling regime. In the strong-coupling regime, the code is described in terms of connectivity and certain spectral conditions on the excitatory subnetworks. We also show that these combinatorial codes are \emph{convex}, that is, they are compatible with patterns of overlaps of convex receptive fields that are common in many sensory systems of the brain \cite{cruz2019open}.

This paper is organized as follows. In \cref{section:Dale_networks} we introduce the necessary background for threshold-linear networks and Dale's law and define the combinatorial and stable combinatorial codes. In \cref{section:main_results} we state our main results. The proofs are relegated to \cref{section:proofs}.  \cref{section:appendix_A} contains some necessary results concerning the stability of fixed points of threshold-linear networks.

\section{Preliminaries: The combinatorial codes of Dale networks}
\label{section:Dale_networks}
 We consider a  standard  firing rate model of a recurrent neural network of $n$ neurons, 
where the dynamics of the  firing rates $x_i(t)\geq 0$   is  described by the differential equations 
\begin{equation}
\label{eq:network_dynamics}
    \dot{x}_i=-x_i+\left[\sum_{j=1}^{n}W_{ij}x_j+b_i\right]_+, \quad i=1,\dots, n,
\end{equation}
and  $\left [ y\right]_+=\max(0,y)$ is the ReLU transfer function.

We assume that this excitatory-inhibitory network respects \emph{\new{Dale's} law} \cite{eccles1976electrical}, \new{where} the neurons are either excitatory (denoted as $\exc$) or inhibitory (denoted as $\inhib$), and the synaptic weights satisfy the following sign constrains:
\begin{align*} \text{excitatory synapses: } i\in \exc  &\implies W_{\new{ji}}\geq 0,  \quad \forall j=1\dots n  \\
 \text{inhibitory  synapses: } i\in \inhib  &\implies W_{\new{ji}}\leq 0,  \quad \forall j=1\dots n .
\end{align*} 
Furthermore, we assume that the diagonal entries of $W$ are zeros,  $W_{ii}=0$ and denote the collection of all $n\times n$ Dale matrices by $\mathbb{D}_n$.

\begin{wrapfigure}{r}{.35\linewidth} 
    \includegraphics[scale=0.15]{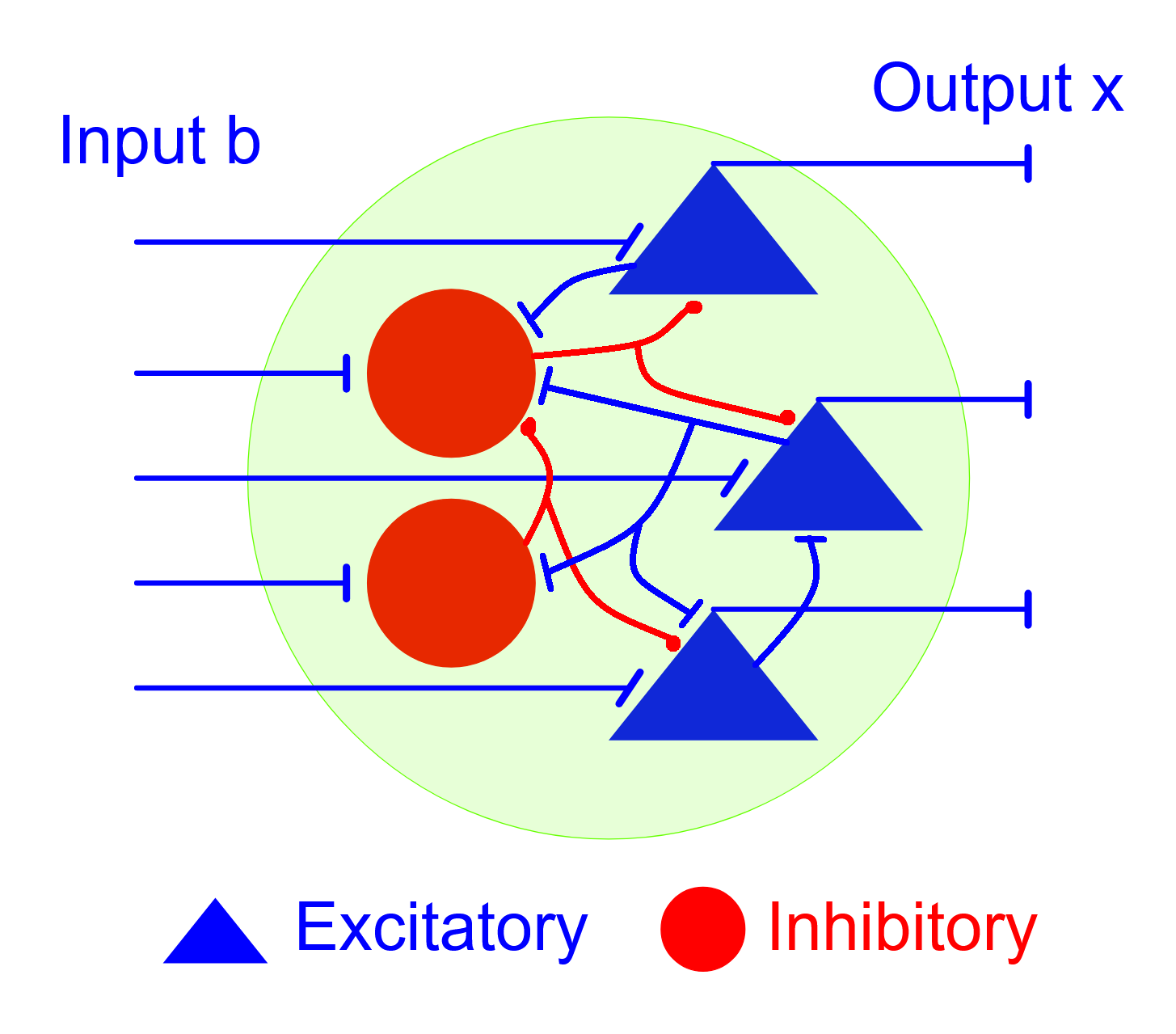}
    \caption{A Dale recurrent network.}
    \label{fig:Dale_network}
\end{wrapfigure}
Following a common architecture of the neocortex, we also assume  that the excitatory neurons ``broadcast'' the output, while the activity of the inhibitory neurons is not observable {\it directly} outside of the network.  We thus consider the setup where the inputs to the network are excitatory, while  only the excitatory output    can be ``read'' from the network (\cref{fig:Dale_network}).  We shall call such a network \emph{a Dale network}. 

We are  interested in the combinatorics of the excitatory output of Dale networks. A combinatorial code    is the set of possible patterns  of neural activation at the fixed points (or steady states). 
For a    firing rate vector $x\in \mathbb R^n_{\geq 0}$,  we     consider the \emph{excitatory support}, i.e. the set of active excitatory neurons: 
$$\excsup {x}=\{i \in \mathcal{E} \, \vert \, x_i>0\}\subset \mathcal{E} .$$  
Recall that  $x^*\in \mathbb{R}^n$ is a fixed point of a network \eqref{eq:network_dynamics} if $x(t)=x^*$ is a constant solution. 
For a given  Dale network \eqref{eq:network_dynamics}, we denote  the set of excitatory supports of all the possible fixed points  as 
$$\new{FP}_+(W,b)=\left\{\supp_+ \, x^* \left \vert \,\,  x^*\in \mathbb R^n_{\geq 0}  \text{ is a fixed point of  \eqref{eq:network_dynamics}}\right. \right\}.$$
Here the  plus sign  highlights  the difference from a somewhat different  definition in \cite{morrison2024diversity}, which considers the combinatorics of all fixed points, which were previously investigated.

 The \emph{combinatorial code} of a Dale synaptic matrix $W$ is the collection of all possible  excitatory support \new{of} all  fixed points in response to all possible inputs: 
 \begin{equation} \label{eq:combinatorial:code} \exccode{W}\od \bigcup_{b\in \mathbb{R}^n_{\ge 0}} FP_+(W,b).\end{equation}
 
The \emph{stable combinatorial code} is the set of the excitatory supports of asymptotically stable fixed points:
\[\stbcode{W}\od\bigcup_{b\in \mathbb{R}^n_{\ge 0}}\{\excsup{x^*}\,|\,x^*\in \mathbb{R}^n_{\ge 0} \textnormal{ is an asymptotically stable fixed point \new{of } \eqref{eq:network_dynamics}}\}.\] 

It turns out that the  combinatorial codes  of  Dale networks can be completely described  in terms of connectivity and spectral radius of the synaptic matrix.  We also show that the combinatorial code is always intersection complete, which implies it is also always a convex combinatorial code \cite{cruz2019open}.

\section{The main results}
\label{section:main_results}
Here we state  the main results, while all the proofs of the theorems are provided in Section \ref{section:proofs}. 
To simplify the mathematical analyses of the network, we make the following mild assumption. 
\begin{assumption}
\label{assumption:switching_assumption} 
The synaptic matrix $W$ of the network \eqref{eq:network_dynamics} 
satisfies the condition that for every non-empty subset of neurons $\sigma\subset[n]$ the principal submatrix $(I-W)_{\sigma}$ is \new{nonsingular}.
\end{assumption}
Note that for a asquare $n\times n $ matrix $A$ and a subset $\sigma\subset[n]$ we denote by $A_\sigma$ the appropriate principal submatrix.
\label{lab:ground:assumption}
\medskip  \noindent In all our results the Ground Assumption is always implicitly assumed.  Note that the set of matrices that do \emph{not} have this property has measure zero, thus this assumption is generically satisfied in any network without fine-tuning. 

\subsection{The role of excitatory-inhibitory  connectivity in shaping  the combinatorial code}
We first observe that to understand the combinatorial code, one can  streamline the excitatory-inhibitory connectivity to its ``essential features'' as follows. 

\begin{theorem}
\label{theorem:replacing_all_inhibitory_neurons_by_one}
Let $W\in \mathbb{D}_{n}$ be a Dale matrix with a set of excitatory neurons $\exc$ and inhibitory neurons $\inhib$. Let $m=|\exc|+1$ and let $W'\in \mathbb{D}_{m}$ be any Dale matrix such that $W'_{\exc}=W_{\exc}$ and for all $i\in \exc$, 
\begin{equation*}  W'_{im}=\begin{cases}
-1, & \exists  j\in \inhib\,   \text{with } W_{ij}\neq 0,\\
0, & \forall j\in \inhib   \,  \text{with }    W_{ij}=0.
\end{cases}
\end{equation*} 
Then $\exccode{W}=\exccode{W'}$.
\end{theorem}

In other words, the combinatorial code $\exccode{W}$  remains unchanged if we replace all the inhibitory neurons  with a single inhibitory neuron  that mimics the connectivity of the entire inhibitory population to each excitatory neuron. Note that here the numerical values of the inhibitory-excitatory weights of the synaptic matrix $W$  do not affect the combinatorial code, even though they may determine the stability of the appropriate fixed points. Furthermore, the following result  states that the connections from excitatory to inhibitory neurons have no influence over the combinatorial code $\exccode{W}$, while they still may determine the stability of the appropriate fixed points, as well as other dynamical properties.

\begin{theorem}
\label{theorem:replace_inhibitory_rows_by_0} Let $W\in \mathbb{D}_{n}$ be a Dale matrix with a set of excitatory neurons $\exc$ and inhibitory neurons $\inhib$, and 
 let $W'\in \mathbb{D}_{n}$ be such that 
\[W'_{ij}=\begin{cases}
0, &\forall i\in \inhib,\, \forall j\in \exc\\
W_{ij}, & \emph{otherwise}
\end{cases}.\]
Then $\exccode{W}=\exccode{W'}$.
\end{theorem}

\bigskip 
\subsection{Characterization of the  combinatorial code of a Dale network}
\cref{theorem:replacing_all_inhibitory_neurons_by_one,theorem:replace_inhibitory_rows_by_0} imply that only $excitatory- excitatory$ synapses and a particular feature of the inhibitory-excitatory connectivity play a role in shaping the combinatorial code. To describe these features we make the following definitions. 
\begin{definition}
\label{def:matrix_graph}
Let $W $ be a Dale matrix, and $\exc$ and $\inhib$  denote the set of its excitatory and inhibitory neurons respectively. 
\begin{itemize}[left=0pt]
\item  An excitatory connectivity graph of $W$ is  a directed graph   $G_\exc $ whose vertices are the exitatory neurons $\exc$,  and whose  arcs   are defined as  
$$ i\to j \iff  W_{ji}\new{>} 0,   \quad \text{where }  i\neq j\in \exc.$$ 
\item The \new{uninhibited} set $\new{\excuninhib}\subset \exc$ is the subset of excitatory neurons that \emph{do not} receive any \new{inhibition}, i.e. 
\[\new{\excuninhib}\od \left\{j\in {\exc}\,\,|\,\, \forall i\in {\inhib},\, W_{ji}=0\right\}.\] 
\item \new{The inhibited set $\new{\excinhib}\subset \exc$ is the subset of excitatory neurons that receive inhibition, i.e. $\excinhib=\exc\setminus \excuninhib$.}
\item For a directed graph $G$  on a set of vertices $\exc$  and a subset $\new{\excuninhib}\subset \exc$,  we define a code of a pair $\left(G, \new{\excuninhib}\right)$ as 
\begin{equation*}
\code{G}{\new{\excuninhib}}  \od \left \{ \sigma\subset\exc \, \vert\,  \new{\outnbhd{G}{\sigma}}\cap \new{\excuninhib} \subset \sigma \right \}\new{=\{\sigma\subset \exc\,|\, \outnbhd{G_{\exc}}{\sigma}\setminus \sigma\subset\mathcal{E_I}\}},
\end{equation*} 
where $$ \new{\outnbhd{G}{\sigma}} \od  \bigcup_{i\in \sigma }\left\{ j \in \exc \, \vert \,  i\to j  \right \}$$ 
\new{denotes the {\it out-neighborhood} of a set of vertices $\sigma$ in the graph $G$. We  consider an excitatory connectivity graph of some Dale network $G=G_{\exc}$ and refer to  the set of excitatory neurons $\outnbhd{G_{\exc}}{\sigma}$  as {\it the synaptic targets}  of a subset   $\sigma$.}
\end{itemize} 
\end{definition}

\noindent \new{These connectivity features are illustrated on an example in  \cref{fig:connectivity_graph}.}
\begin{figure}[H]
\center 
    \includegraphics[scale=0.6]{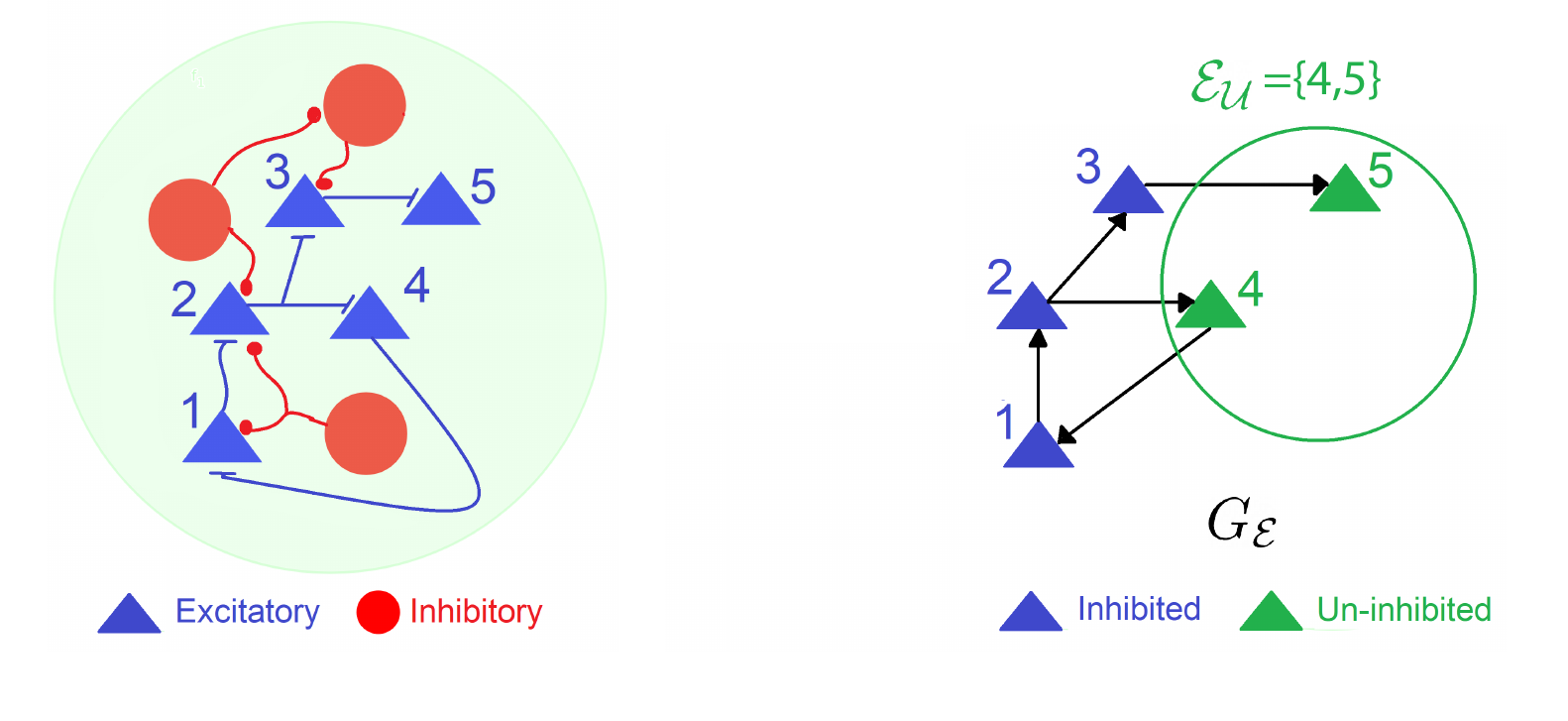}
    \caption{A Dale network (left) and its  connectivity features  (right).}
    \label{fig:connectivity_graph}
    \end{figure}
\noindent \new{Intuitively,  a network has  a fixed point, supported at excitatory neurons $\sigma\subset \exc$, if  all the other excitatory neurons are  silent. This will not occur if they receive excitatory input from neurons in $\sigma$ and don’t have any inhibitory neurons to silence  them.  This intuition is formalized in the definition of   
$\code{G_{\exc}}{\new{\excuninhib}}$.}

\bigskip 
The following theorem describes the combinatorial code of a Dale network. 

\begin{theorem}
\label{theorem:positive_code_characterization}
Let $W $ be a Dale matrix, and $\sigma\subset \exc$ be a non-empty subset of excitatory neurons.   
Then  $  \sigma\in 
\exccode{W}$ if and only if the following two conditions are both satisfied:
\begin{enumerate}
   \item[(i)]  (the spectral condition) $\rho(W_{\new{\excuninhib\cap\sigma}})<1$, 
\item[(ii)] (the graph condition) $\sigma\in \code{G_{\exc}}{\new{\excuninhib}}$,
\end{enumerate}
where $W_{\new{\excuninhib\cap\sigma}}$ denotes the synaptic weights of the excitatory sub-network on the subset $\new{\excuninhib\cap\sigma}$, and  $\rho(W_{\new{\excuninhib\cap\sigma}})$ denotes the spectral radius of the matrix $W_{\new{\excuninhib\cap\sigma}}$.
\end{theorem}

 \begin{figure}[H]
 \center 
    \includegraphics[scale=0.6]{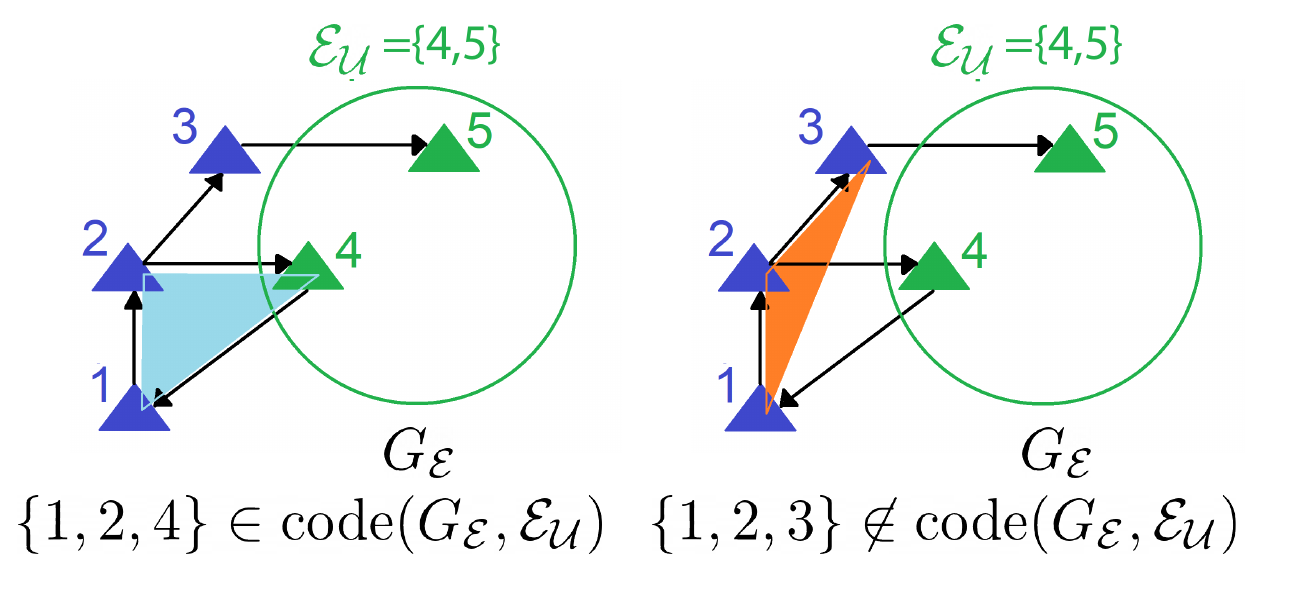}
    \caption{An excitatory connectivity graph of a Dale network. A codeword in $\code{G_{\exc},\excuninhib}$ (left) and a codeword not in $\code{G_{\exc},\excuninhib}$ (right).}
    \label{fig:codewords}
    \end{figure} 
\new{To illustrate an application of \cref{theorem:positive_code_characterization}, consider the excitatory connectivity graph of a Dale network from \cref{fig:connectivity_graph}. 
From \cref{fig:codewords}, we see that  $\outnbhd{G_{\exc}}{\excuninhib}(\{1,2,4\})=\{1,2,3,4\}$ and  $\outnbhd{G_{\exc}}{\excuninhib}(\{1,2,3\})=\{2,3,4,5\}$. Thus
\[\outnbhd{G_{\exc}}{\excuninhib}(\{1,2,4\})\cap\excuninhib=\{4\}\subset \{1,2,4\}\]
and 
\[\outnbhd{G_{\exc}}{\excuninhib}(\{1,2,3\})\cap\excuninhib=\{4,5\}\not\subset \{1,2,3\}.\]
Thus, since $\{1,2,4\}\in \code{G_{\exc},\excuninhib}$, we have a {\it candidate} for an element in $\exccode{W}$. Assuming one has access to the synaptic weights, the last remaining  step would be  checking the spectral radius condition to see if $\{1,2,4\}$ is indeed in the code.}

\bigskip 
We say that the network \eqref{eq:network_dynamics}
is \emph{weakly coupled} \new{ \cite{van1996chaos,hoppensteadt2012weakly,mastrogiuseppe2018linking}} if the Frobenius matrix norm 
$$||W||_F=\sqrt{\text{trace}(W^TW)}$$ 
of its synaptic matrix $W$ is smaller than $1$.
It is natural to consider the weak coupling regime separately, as in this regime the spectral condition in \cref{theorem:positive_code_characterization}  is always satisfied.

 \begin{theorem}
\label{theorem:excitatory_code_stability}
Let $W$ be a nonsingular Dale matrix that is weakly coupled, i.e.  $||W||_{F}<1$. Then every fixed point is \new{asymptotically stable}, and the combinatorial code is completely described by the graph condition: 
\[\mathcal{C}(W)=\mathcal{SC}(W)=\code{G_{\exc}}{\new{\excuninhib}}.\] Furthermore, for all $b\in \mathbb{R}^n_{\ge 0}$ there is a unique globally exponentially stable fixed point of \eqref{eq:network_dynamics}.
\end{theorem}

 \bigskip 
The above theorem  implies that the combinatorial code of weakly coupled networks is completely determined by the connectivity features  \emph{alone}, and does not depend on the strengths of the synapses as long as the network is in the weak coupling regime.  
The following cautionary example illustrates two Dale matrices that are \emph{not} weakly coupled. These matrices have the same connectivity, but they exhibit different eigenvalue spectra of excitatory subnetworks, resulting in different combinatorial codes.
Consider Dale matrices  
\begin{equation*}\label{example:positive_code_counter_example}
\begin{small} W=\begin{pmatrix}
0 & 2 & 0 & -2\\
1 & 0 & 2 & 0\\
0 & 1 & 0 & 0\\
1 & 1 & 1 & 0
\end{pmatrix}, \end{small}  \quad \text{and}   \quad \begin{small}U=\begin{pmatrix}
0 & 3 & 0 & -2\\
0.5& 0 &0.5& 0\\
0 & 0.5 & 0 & 0\\
1& 1 & 1 & 0
\end{pmatrix}
\end{small}
\end{equation*}
and note that both these matrices have the same excitatory connectivity graph $G_{\exc}$ and the same excitatory neurons  $\new{\excuninhib}=\{2,3\}$ that do not receive inhibition. Let $\sigma=\{1,2,3\}$, thus $\new{\excuninhib\cap\sigma}=\{2,3\}$ and observe that $\rho(W_{\new{\excuninhib\cap\sigma}})=\sqrt{2} >1$,  thus by \cref{theorem:positive_code_characterization} $\sigma\not\in \exccode{W}$. On the other hand $\rho(U_{\new{\excuninhib\cap\sigma}})=0.5<1$. Furthermore, because  $\sigma=\exc$,   the graph condition  $\new{\outnbhd{G_{\exc}}{\sigma}}\cap \new{\excuninhib}\subset\sigma$ is satisfied and thus by \cref{theorem:positive_code_characterization}, $\sigma\in \exccode{U}$.

\subsection{Combinatorial  codes of  Dale recurrent networks are convex}

 Perhaps the most surprising  implication of  Theorem \ref{theorem:positive_code_characterization}  is that Dale networks naturally produce convex combinatorial codes.  To explain the relevant background, we first motivate and define  convex combinatorial codes. 
 \\
 
  There are two complimentary viewpoints on what determines neural activity in sensory systems.  One viewpoint is that  the  brain represents information via patterns of neural activity that arise as a result of neural dynamics.  A different perspective is that  the neural activity in sensory neural systems is induced by external stimuli, whereby each neuron responds to a given stimulus according to its own receptive field.  \new{If these two views are both correct, then the patterns of neural activity resulting from the dynamics should  be compatible with those determined by the intersection patterns of the receptive fields.} 
\\

We define a receptive field of an individual neuron as a subset $U\subset \mathbb R^d$ in  a stimulus space $\mathbb R^d$, such that the firing rate $x(t)$ of a given neuron is non-zero at the times when  the stimulus is in the region $U$.  Given a collection of receptive fields $\mathcal{U}=\{ U_i\}$ of  a population of neurons $\exc$,  consider \emph{the code} of $\mathcal{U}$, that describes all possible intersection patterns of the receptive fields $U_i$ as   
\[\text{code}({\mathcal{U}}) \od   \{\sigma\subset \exc\,|\, R_{\sigma}\neq \varnothing \},\]
where $$R_{\sigma}\od\left(\bigcap_{i\in \sigma}U_i\right)\setminus\bigcup_{j\not\in \sigma}U_j$$
 denotes the region in the stimulus space where each of the neurons in $
\sigma$  is activated and no other considered neuron is active, and $R_\varnothing\od \mathbb R^d\setminus \bigcup_{j\in \exc}U_j$. The regions $R_\sigma$ partition the stimulus space $\mathbb R^d$ (see \cref{ex:codeexample}). 

\begin{figure}[H]
\begin{center}
\definecolor{qqccqq}{rgb}{0,0.8,0}
\definecolor{qqaaqq}{rgb}{0,0.6,0}
\definecolor{qqffqq}{rgb}{0,1,0}
\definecolor{qqqqzz}{rgb}{0,0,0.6}
\definecolor{rrqqqq}{rgb}{0.8,0,0}
\definecolor{ffqqqq}{rgb}{1,0,0}
\definecolor{qqqqaa}{rgb}{0,0,0.4}
\definecolor{qqzzzz}{rgb}{0,0.6,0.6}
\definecolor{qqrrrr}{rgb}{0,0.4,0.4}
\begin{tikzpicture}[line cap=round,line join=round,x=1.0cm,y=1.0cm, scale=1]
\draw [line width=0.4pt,dash pattern=on 3pt off 3pt,color=qqccqq,fill=qqccqq,fill opacity=0.25] (-2.0,2.9) ellipse (2.3cm and 1.15cm);
\draw [line width=0.4pt,dash pattern=on 3pt off 3pt,color=qqqqzz,fill=qqqqzz,fill opacity=0.25] (1.26,3.04) ellipse (2.3cm and 1.15cm);
\draw [line width=0.4pt,dash pattern=on 3pt off 3pt,color=rrqqqq,fill=rrqqqq,fill opacity=0.25] (0.15,3.1) ellipse (2.5cm and .45cm);
\draw [line width=0.4pt,dash pattern=on 3pt off 3pt,color=qqzzzz,fill=qqzzzz,fill opacity=0.25] (1.1,2.8) ellipse (0.35cm and 2.0cm);
\begin{scriptsize}
\definecolor{ttfftt}{rgb}{0.2,0.8,0.2}
\definecolor{qqqqff}{rgb}{0,0,1}
\definecolor{ffqqtt}{rgb}{1,0,0.2}
\definecolor{uququq}{rgb}{0.25,0.25,0.25}
\definecolor{zzttqq}{rgb}{0.6,0.2,0}
\definecolor{ffqqff}{rgb}{1,0,1}
\draw[color=ffqqqq] (-1.9,3.1) node {$U_1$};
\draw[color=qqqqff] (0.4,2.3) node {$U_3$};
\draw[color=qqaaqq] (-3.32,2.34) node {$U_2$};
\draw[color=qqrrrr] (1.1,1.25) node {$U_4$};
\draw (0.25,5) node{$R_{13}$};
\draw[->] (0.25,4.8) -- (0.45,3.1);
\draw[->] (0.25,4.8) arc (65:25:3.5cm);
\end{scriptsize}
\end{tikzpicture}
\caption{An example of a collection of  receptive fields $\mathcal{U}=\{U_1,U_2,U_3,U_4\}$ and its code, $\mathcal{C}=\textnormal{code}(\mathcal{U})= \{\varnothing, 2, 3, 4, 12, 13, 23, 34, 123, 134\}$.  Here we denote a codeword $\{i_1, i_2, \dots, i_k\} \in \mathcal{C}$ by the string $i_1i_2\dots i_k$;  for example, $\{1,3,4\}$ is abbreviated to $134$. 
Note that  all $U_i$ are convex (and thus connected), but the regions $R_\sigma$ are typically non-convex for non-maximal subsets $\sigma $. For example $R_{13}$ is neither convex nor connected.}
\label{ex:codeexample}
\end{center}
\end{figure}
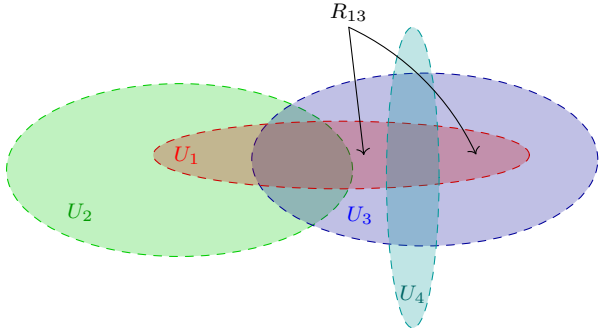

Note that $\text{code}({\mathcal{U}}) $ is the set of all possible patterns of neuronal activation that are compatible with  the receptive fields $\mathcal{U}$. The same patterns should be allowed by the neural network dynamics. 
\new{We associate a combinatorial neural code $\mathcal{C}(W)$ to a Dale network  \eqref{eq:network_dynamics}, and  we interpret the fixed points of the excitatory neurons, {\it both stable and unstable}, as 
our  model for transiently activated stimulus representations. This is because the neural activity spends a significant amount of time not only in the neighborhood of stable fixed points, but also  around saddle fixed points, while  evolving along heteroclinic trajectories (see e.g. \cite{huerta2004reproducible,ashwin2005instability,rabinovich2011robust}).}  
\new{Following these interpretations we posit that the combinatorial 
code of the network  must  be the same as  the intersection patterns of the receptive fields:  
 \[\mathcal{C}(W)=\text{code}({\mathcal{U}}).
\]
}
\bigskip 
Recall that a set $U\subset \mathbb R^d$ is called \emph{convex} if for any two points $x,y\in U$ the line segment $[x,y]$ also lies in $U$. 
A number  of sensory areas in the brain possess convex receptive fields; a very incomplete  list of such areas includes the   hippocampus, the primary visual cortex, the primary auditory cortex, etc.   Following  \cite{CIVY13,cruz2019open,hyperplane2020,curto2017makes}, we 
define  convex combinatorial codes as follows. 

\begin{definition}
\label{def:convex_codes}
A combinatorial code $\mathcal{C}\subset 2^{[n]}$ is \emph{open convex} 
if $\mathcal{C}=\text{code}(\mathcal{U})$  for a collection $\mathcal{U}=\{U_i\}$ of open convex 
subsets in a Euclidean space $ \mathbb{R}^d$, for some $d\ge 1$.   
\end{definition}

\begin{wrapfigure}{r}{.28\linewidth} 
\centering
\includegraphics[ trim = 0.0in 0.0in 0.00in 0.0in, width=1.0in]{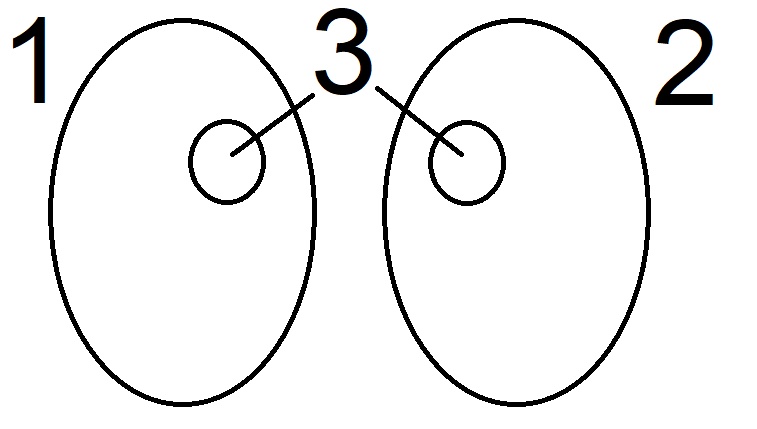}
 \vspace{.12in}
\caption{ A  non-convex  code $\mathcal{C}=\{1, 2, 13, 23\}$. 
}
\label{fig:nonconvex}
\end{wrapfigure}
It has been  previously established in  \cite{CIVY13, giusti2014no, cruz2019open} that not every combinatorial code is open convex.  The smallest example of a non-convex code is illustrated in \cref{fig:nonconvex}. The topological properties that prevent a combinatorial code from being convex  were studied in \cite{cruz2019open,curto2017makes,LienkaemperShiuWoodstock}. In fact, a randomly chosen\footnote{This, of course, requires a proper definition of the probability distribution on the set of all codes that we omit here.} combinatorial code on a large number of neurons is non-convex with a high probability. 

Following \cite{cruz2019open}, recall that a combinatorial code $\mathcal{C}$ is called \emph{intersection-complete}, if \new{the} intersection of any two codewords in $\mathcal{C}$ also belongs to  $\mathcal{C}$: 
$$ \sigma,\nu\in \mathcal{C} \implies \sigma\cap \nu \in \mathcal{C}. $$
If a code is intersection-complete and  also has the property that a union of two codewords is a codeword, then such code is   a \emph{sublattice} of $2^\exc$ \cite{roman2010lattices}. 
It has been  previously established in \cite{cruz2019open} that  intersection-complete codes are open convex.

\begin{theorem}{\cite{cruz2019open}}
\label{theorem:convex_codes}
\new{Every intersection-complete code $\mathcal{C}\subset 2^{[n]}$ is open convex.}
\end{theorem}

\medskip 

We use this result to show that the combinatorial codes of a Dale network are convex. 
Consider a directed graph  $G_\exc$, whose  vertices are the excitatory neurons $\exc$, and whose edges are derived from the   the excitatory connectivity as in  \cref{def:matrix_graph}.  
First we observe the following.

\begin{proposition}
\label{proposition:graph_rules}
The code $\code{G_\exc}{\new{\excuninhib}}$ is a sublattice of the Boolean lattice $2^\exc$.
\end{proposition}

This observation, combined with \cref{theorem:excitatory_code_stability}  implies  that the combinatorial code of a weakly coupled network is a sublattice. If the network is \emph{not} weakly coupled, then the spectral condition in  \cref{theorem:positive_code_characterization} may prevent the code of the network from being a sublattice, however the code $\exccode{W}$ of a Dale network  remains  intersection-complete. 

\begin{theorem}
\label{theorem:positive_code_is_convex}
For a Dale synaptic matrix $W$, its code  $\exccode{W}$ is intersection-complete and is thus convex.
\end{theorem}

The proof of this theorem is given in \cref{sec:the:proofs}. 
We suggest that this theorem may provide a new explanation for the prevalence of convex receptive fields of excitatory neurons in recurrent circuits of sensory systems.

\subsection{Constructing a Dale network from a combinatorial code}
Can one  encode an arbitrary combinatorial code as the combinatorial code of a Dale network? \cref{theorem:positive_code_is_convex}  implies that the code needs to be intersection-complete, to be realized on a generic Dale network. In addition, \cref{proposition:graph_rules} tells us that  if one wants to build a weakly coupled network with a prescribed code, this code needs to be a sublattice. 
It turns out that any  sublattice can be encoded on a weakly coupled Dale network. 

\begin{theorem}
\label{theorem:constructing_network} Given a sublattice ${\C}\subset 2^{\exc}$ with $\varnothing,\exc\in \C$, 
 define  a map $c:2^{\exc}\to \C$ as 
\begin{equation}\label{eq:cmap}  
c(\sigma)\od \bigcap_{{\nu\in \C},\atop {\nu\supset \sigma}} \nu,
\end{equation} 
and consider a directed graph $G_c$, whose vertices are $\exc$ and whose edges are  defined via the rule 
\begin{equation}\label{eq:Gc}  \quad i\to j \iff  j\in c(\{i\}) \text{ and }  i\neq j.\end{equation} 
Then  $\C=\code{G_c}{\exc}$.
\end{theorem}
\noindent \new{Note that $\code{G_c}{\exc}$ from \cref{theorem:constructing_network} is realized by a network with no inhibition.}
This theorem, combined with \cref{theorem:excitatory_code_stability} translates into the following 
\begin{corollary} \label{cor:sublattice:encoding} Given a sublattice $\C\subset 2^{\exc}$ with $\varnothing,\exc\in \C$, one can find a weakly coupled network $W$ of  excitatory  neurons  with $\C=\mathcal{C}(W)$. 
\end{corollary}

\new{This} ``\emph{learning rule}" for an excitatory  synaptic matrix $W$  amounts to  assigning small non-zero excitatory synaptic weights according to the directed arcs in the graph $G_c$ defined in equations  \eqref{eq:cmap},\eqref{eq:Gc}.

The following example  illustrates  the procedure described above to construct a graph of connectivity from  a given code. 
Consider a sublattice  ${C}=\{\varnothing,\{4\},\{2,4\},\{3,4\},\newline \{2,3,4\},\{1,2,3,4\}\}$. \new{One can directly compute that $c(\{1\})=\{1,2,3,4\}$, $c(\{2\})=\{2,4\}$, $c(\{3\})=\{3,4\}$ and $c(\{4\})=\{4\}$.}
It is \new{then} straightforward to verify that $G_c$ is the graph in  \cref{fig:constructing_network} and thus $ {C}=\code{G_c}{\{1,2,3,4\}}$.  

\begin{figure}[H]
  \begin{center}
     \begin{tikzpicture}[scale=0.5]       
      \node[label={[xshift=-0.2cm, yshift=-0.3cm]$1$},shape=circle,fill=red!60,draw=black,scale=0.375] (A) at (1,-4) {};          
    \node[label={[xshift=0.2cm, yshift=-0.3cm]$2$},shape=circle,fill=red!60,draw=black,scale=0.375] (B) at (4,-4) {};
     \node[label={[xshift=-0.2cm, yshift=-0.3cm]$3$},shape=circle,fill=red!60,draw=black,scale=0.375] (C) at (1,-1) {};
      \node[label={[xshift=0.2cm, yshift=-0.3cm]$4$},shape=circle,fill=red!60,draw=black,scale=0.375] (D) at (4,-1) {};           
           \path [->,line width=1](A) edge node[left] {} (B);
     \path [->,line width=1](C) edge node[left] {} (D);
     \path [->,line width=1](A) edge node[left] {} (C);
     \path [->,line width=1](B) edge node[left] {} (D);
         \path [->,line width=1](A) edge node[left] {} (D); 
    \end{tikzpicture}
    \end{center}
    \caption{The graph $G_c$ on four vertices with   $\code{G_c}{[4]}=\{\varnothing,\{4\},\{2,4\},\{3,4\},\{2,3,4\},[4]\}$.}
        \label{fig:constructing_network}
\end{figure}

\begin{remark}
\label{remark:constructing_network}
\new{
This construction is common and known in finitely generated topologies, or equivalently \emph{Alexandroff topological spaces} \cite{Alexandroff:1937}. Note that $\varnothing,\exc\in C$ and $C$ being a lattice, means that $C$ is a collection of closed sets of a topology on $\exc$. In particular, the code $\mathcal{C}$ are the closed sets of a topology on a finite set. The digraph we obtained is also known as the \emph{specialization preorder}\footnote{Technically, the digraph we would obtain would have all self-loops as it is a preorder and hence a reflexive relation on the finite set. However, we ignore the loops in our construction which is not a problem  in our context.}, associated to the topology of $\C$.}
\end{remark}

\section{Discussion} 
We investigated threshold-linear recurrent networks that satisfy \new{Dale's} law and established  what features of the synaptic connectivity are responsible for determining their combinatorial codes. \cref{theorem:positive_code_characterization} describes the combinatorial code in terms of the connectivity graph and the spectral radii of the appropriate sub-matrices. In the case of weakly coupled networks, the spectral conditions are always satisfied and the combinatorial code depends only on the features in $\code{G}{\new{\excuninhib}} $ that are derived from the connectivity graph 
(\cref{theorem:excitatory_code_stability}). In this situation, all the fixed points are stable, moreover the combinatorial code is a sublattice. 
We also proved that any \new{sublattice} can be encoded as a combinatorial code of an excitatory network (\cref{cor:sublattice:encoding}). 

If the network is \emph{not} weakly coupled, the spectral condition in  \cref{theorem:positive_code_characterization} may prevent the combinatorial code from being a sublattice. Intuitively,  this is because the union of two 
codewords in $\exccode{W}$ may not be in $\exccode{W}$, as the spectral radius of the appropriate larger matrix may (or may not)  increase.  We have proven in \cref{theorem:positive_code_is_convex} that every combinatorial code is intersection-complete and therefore convex. Can any intersection-complete code can be encoded on a Dale network?  The answer to this question is currently unknown, as it requires a better understanding of the interplay of the spectral and the graph conditions in \cref{theorem:positive_code_characterization} .

Finally, we hypothesize that the result (\cref{theorem:positive_code_is_convex}) that the combinatorial codes of Dale recurrent networks are always convex  may provide a natural explanation for the ubiquity of  convex receptive fields in  many sensory systems in the mammalian brain.

\section{Proofs}
\label{section:proofs} 

In this section we \new{give} proofs for our main results stated in \cref{section:main_results}.  
First, we \new{recall}  the conditions for having a fixed point of the dynamics of \eqref{eq:network_dynamics}.
\begin{lemma}{\cite[Proposition 2.1]{curto2012flexible}}
\label{lemma:fixed_point_conditions} For a threshold linear network in eq. \eqref{eq:network_dynamics}, characterized by $(W,b)$, a point $x^*\in \mathbb{R}^n$ is a fixed point with support $\sigma\subset [n]$ if and only if the following conditions are all satisfied
\begin{enumerate}
    \item $(I-W)_{\sigma}x^*_{\sigma}=b_{\sigma}$,
    \item $x^*_{\sigma}>0$, and
    \item $b_{\overline{\sigma}}\le -W_{\overline{\sigma}\sigma}x^*_{\sigma}$.
\end{enumerate}
\bigskip 

Here $b_\sigma$ \new{is} the restriction of a  vector $b$ to the indices in $\sigma$, $\overline{ \sigma}$ denotes the complement in $[n]$ to a subset $\sigma$,  $W_{\overline{\sigma}\sigma}$ is the rectangular submatrix of $W$, restricted to the  rows in $\overline{\sigma}$ and columns in $\sigma$, and $(I-W)_\sigma=(I-W)_{ \sigma\sigma}$ is the principal submatrix restricted to the subset $\sigma$. 
Note that if $(I-W)_{\sigma}$ is invertible, then there is at most one fixed point with support $\sigma$ and it is given by $x^*_{\sigma}=(I_{\sigma}-W_{\sigma})^{-1}b_{\sigma}$.
\end{lemma}

\new{The conditions for having a fixed point in \cref{lemma:fixed_point_conditions} follow (with some work) from the dynamical system in \eqref{eq:network_dynamics}. In particular, $(I-W)_{\sigma}x_{\sigma}^*=b_{\sigma}$ follows from needing to evaluate the ReLU nonlinearity to a positive number, and $b_{\overline{\sigma}}\le -W_{\overline{\sigma}\sigma}x_{\sigma}^* $ follows from evaluating the ReLU nonlinearity to $0$.}

A key ingredient in our  proofs is the observation that the first condition in \cref{lemma:fixed_point_conditions} can be translated to the language of semipositive matrices (\cref{def:semipositive_matrix} below). From there we make use of the  theory   $M$-matrices (\cref{def:M-matrix} below). Note that the machinery of semipositive matricies has been previously used   in a different  context of recurrent networks with Heaviside  transfer function  in \cite{TournoyDoiron}.

\subsection{Necessary matrix theory results}
First, we recall some definitions and results from the theory of $M$-matrices, following  \cite{berman1994nonnegative}. An ${m\times n}$ matrix  $A$ is \emph{positive} (\emph{nonnegative}) if $A_{ij}>0$ ($A_{ij}\ge 0$) for all $ i $ and  $ j$. If $A$ is positive (nonnegative) we denote this by $A>0$ ($A\ge 0$).  
For  a vector $x\in \mathbb{R}^n$, we write $x>0$, $x<0$, $x\ge 0$, or $x\le 0$ if all  the entries of  corresponding  vector $x$ satisfy the appropriate inequality.

\begin{definition}\cite[Chapter 6]{berman1994nonnegative}
\label{def:M-matrix}
\new{A} square matrix $A$ is a \emph{$Z$-matrix} if $A_{ij}\le 0$ for all $i\neq j$. A $Z$-matrix $A$ is called an \emph{$M$-matrix} if  $A=sI-B$, where $I$ denotes the identity matrix,  $B_{ij}\ge 0$ for all $i\neq j$,  and the scalar $s$ is not smaller than the spectral radius of the matrix $B$,  $s\ge \rho(B)$.
\end{definition}

\medskip 
\noindent The following \new{lemma} will be used later. 

\begin{lemma}{\new{\cite[Lemma 2.1]{berman1994nonnegative}}}
\label{lemma:simple-radius-lemma} Assume that $A=I-W$, where $W\geq 0$. Then the following two statements are equivalent: 
\begin{enumerate}
\item[(i)] The matrix $A$ is a \new{nonsingular} $M$-matrix.
 \item[(ii)] The spectral radius of $W$  is smaller than one: $\rho(W)<1$.
\end{enumerate} 
\end{lemma}

\medskip 
\noindent We shall also make use of the following. 
\begin{lemma}{\cite[Chapter 2, Corollary 1.6]{berman1994nonnegative}}
\label{lemma:eigenradius_of_submatrix}
Let $A$ be a square nonnegative matrix. Suppose that $B$ is a principal submatrix of $A$. Then $\rho(B)\le \rho (A)$.
\end{lemma}

\begin{definition}
\label{def:semipositive_matrix}
A square matrix $A$ is \emph{semipositive} if $\exists x>0$ such that $Ax>0$. 
\end{definition}

\begin{lemma}
\label{lemma:invertible_semipositive_characterization}
Let $A$ be a \new{nonsingular} $n\times n$ matrix. Then $A$ is semipositive if and only if there exists a vector  $ x>0$ such that $Ax\ge 0$.
\end{lemma}

\begin{proof}
The forward direction is immediate. For the converse suppose that $\exists x>0$ such that $Ax\ge 0$. Let $a=||A^{-1}||_{op}$ be the operator norm of $A^{-1}$, $x_m=\min_{1\le i\le n}(x_i)$, $\varepsilon=\frac{x_m}{2a\sqrt{n}}$ and $\overline{\varepsilon}\in \mathbb{R}^n$ be the column vector whose all components are equal to $\varepsilon$. Then $||\overline{\varepsilon}||_2=\frac{x_m}{2a}$ by construction. Furthermore,  $||A^{-1}\overline{\varepsilon}||_2\le a||\overline{\varepsilon}||_2=a\frac{x_m}{2a}=\frac{x_m}{2}$. In particular, $\sum_{i=1}^n(A^{-1}\overline{\varepsilon})_i^2\le \frac{x_m^2}{4}$ for all $1\le i\le n$. Therefore, $(A^{-1}\overline{\varepsilon})_i^2\le \frac{x_m^2}{4}$ which implies $|(A^{-1}\overline{\varepsilon})_i|\le \frac{x_m}{2}$ for all $1\le i\le n$. Thus, $x+A^{-1}\overline{\varepsilon}>0$ because of the way $x_m$ was defined. Finally, $A(x+A^{-1}\overline{\varepsilon})=Ax+\overline{\varepsilon}\ge 0 +\overline{\varepsilon}>0$ and thus $A$ is semipositive.
\end{proof}

The following example illustrates  the necessity of the hypothesis that $A$ is a \new{nonsingular} matrix in \cref{lemma:invertible_semipositive_characterization}.
 Consider a singular matrix 
 $A=\begin{pmatrix}
1 & -1 & 0\\
-1 & 1 & 0\\
0 & 0 & 1
\end{pmatrix}$. Let $x=\begin{pmatrix}
1 & 1 & 1
\end{pmatrix}^T$. Then note that $Ax\ge 0$. Now let $y\in \mathbb{R}^3_{>0}$. Suppose that $Ay>0$. This implies that $y_1>y_2$ and $y_2>y_1$ which cannot be. Thus $A$ is not semipositive.

There are  numerous characterizations of $M$-matrices; one can find several dozen conditions  on a matrix  that are all equivalent to being an $M$-matrix \cite{berman1994nonnegative}. However, our proofs  rely only  on the following observation.

\begin{theorem}\cite[Chapter 6, Theorem 2.3]{berman1994nonnegative}
\label{theorem:M_matrix_equivalence}
Let $A$ be a $Z$-matrix. Then $A$ is a \new{nonsingular} $M$-matrix if and only if $A$ is semipositive. 
\end{theorem}

\subsection{\new{Proof of \cref{theorem:replacing_all_inhibitory_neurons_by_one} and related results}}\label{sec:the:proofs}

In order to prove \texorpdfstring{\cref{theorem:replacing_all_inhibitory_neurons_by_one}} a few preliminary results are necessary. Namely \texorpdfstring{\cref{theorem:code_is_invariant_if_adding_inhib_neuron_with_no_extra_connections,corollary:replacing_all_inhibitory_neurons_by_one}}. We also introduce the following notation for convenience.

\begin{definition}
\label{def:set_of_fixed_pt_supports}
For the dynamical system in \eqref{eq:network_dynamics} we denote the set of supports of all the fixed points by $FP(W,b)$ and we let $FP(W)\od\bigcup_{b\ge 0}FP(W,b)$.
\end{definition}

\begin{theorem}
\label{theorem:code_is_invariant_if_adding_inhib_neuron_with_no_extra_connections}
Let $W\in \Dale{n}$ and   $W'\in \Dale{n+1}$ be  Dale matrices such that $W'_{[n]}=W$, and the last  $(n+1)$-st column of $W'$ is inhibitory and  satisfies the following condition: 
\begin{equation}
\label{eq:additional_inhib_neuron}
\exists i\in \exc \, \text{with  }  W'_{i(n+1)}<0 \quad \implies\quad  \exists j \in \inhib \,  \text{with } W_{ij}<0.
\end{equation} 
Then $\exccode{W}=\exccode{W'}$.
\end{theorem}

\begin{proof}
The proof needs that both matrices  $W$  and $W'$ satisfy  the Ground Assumption (page \pageref{lab:ground:assumption}).  Since being \new{nonsingular} is a fine-tuned condition for matrices, this is generically true. Furthermore, we can explicitly construct a $W'$ in such a way that it satisfies the Ground Assumption and the hypothesis in the statement of the theorem. In particular we can choose the $(n+1)$-st row of $W'$ to be the zero vector. This will make it so that for any $\sigma\subset [n+1]$, $(I-W')_{\sigma}$ is \new{nonsingular}, assuming that $W$ was already satisfying the Ground Assumption. We now proceed with the proof. 

\new{Observe that by construction, as we are only adding an additional inhibitory neuron to the network, both $W$ and $W'$ have the same collection of excitatory neurons, $\exc$. Furthermore, $\mathcal{E_I}\subset \mathcal{E_I}'$, where $\mathcal{E_I}'$ are the excitatory neurons in $W'$ that receive inhibition. However, \eqref{eq:additional_inhib_neuron} in fact yields us $\mathcal{E_I}'\subset \mathcal{E_I}$. Therefore $\mathcal{E_I}=\mathcal{E_I}'$, that is to say the collection of inhibited and uninhibited excitatory neurons in $W$ and $W'$ are the same. Hence, there is no confusion when we write the excitatory neurons in $W'$ as a partition $\exc=\mathcal{E_U}\sqcup \mathcal{E_I}$.}

Suppose that $\sigma\in \exccode{W}$. Then, there exists $\tau\subset \inhib$ such that $\nu=\sigma\sqcup\tau\in FP(W)$, and (by Lemma \ref{lemma:fixed_point_conditions})  a vector  $b\in \mathbb{R}^n_{\ge 0}$ such that $x_\nu=(I_{\nu}-W_{\nu})^{-1}b_{\nu}>0$ and $0\le b_{\overline\nu}\le -W_{\overline{\nu} \nu}x_\nu $, where $\overline{\nu}=[n]\setminus \nu$.  Let $\nu'=\nu \new{\sqcup}\{n+1\}$ and  define $y\in \mathbb{R}^{|\nu|+1}_{>0}$ as follows. The first $\vert \nu\vert$ entries of $y$ are those of $x_\nu$, i.e. $y_{\nu}=x_\nu\in \mathbb R^{|\nu|}_{>0}$, and the last entry is defined as  \begin{equation*}   y_{|\nu|+1}=\left \vert \sum_{k\in \nu }(I-W')_{(n+1)k}x_{k}\right \vert +1.
\end{equation*} 
\new{Note that for $i\in \nu$, we have 
\[((I-W')_{\nu'}y)_i=\sum_{k\in \nu'}(I-W')_{ik}y_k=\sum_{k\in \nu}(I-W)_{ik}x_k-W'_{i(|\nu|+1)}y_{|\nu|+1}\ge 0.\]
Thus}, by construction we have $(I-W')_{\nu'}y\ge 0$. Furthermore, $-W'_{\overline{ \nu'} \nu'}y\ge -W_{\overline{\nu} \nu}x_\nu \ge 0$. Let $b'\in \mathbb{R}^{n+1}_{
\ge 0}$ be defined by $b'_{\nu'}=(I-W')_{\nu'}y$ and $b'_{\overline{\nu'}}=0$. Then, by \cref{lemma:fixed_point_conditions}  $\nu'=\sigma\sqcup \tau\cup \{n+1\}\in FP(W',b')\subset FP(W')$,  and thus $\sigma\in \exccode{W'}$.

\medskip 

Now suppose that $\sigma\in \exccode{W'}$. Then there exists a $\tau'\subset \inhib\sqcup \{n+1\}$ such that $\nu=\sigma\sqcup \tau' \in FP(W')$. Since $\nu \in FP(W')$, there exists a $b'\in \mathbb{R}^{n+1}_{\ge 0}$ such that 
\[x\od(I_{\nu}-W'_{\nu})^{-1}b'_{\nu}>0, \textnormal{ and } 0\le b'_{[n+1]\setminus\nu}\le -W'_{([n+1]\setminus\nu)\nu}x,\] 
by \cref{lemma:fixed_point_conditions}. We have two cases to consider; either $\tau'\subset \inhib$ or $(n+1)\in \tau'$. 

In the case that $\tau' \subset \inhib$, observe that $(I_{\nu}-W'_{\nu})^{-1}=(I_{\nu}-W_{\nu})^{-1}$ and that $W'_{([n+1]\setminus\nu)\nu}$ is almost identical to $W_{([n]\setminus\nu)\nu}$; $W'_{([n+1]\setminus\nu)\nu}$ can be obtained from $W_{([n]\setminus\nu)\nu}$ by adding an extra row of entries induced from the $(n+1)$-st row of $W'$. Define $b\in \mathbb{R}^{n}_{\ge 0}$ by $b_{\nu}=(I-W)_{\nu}x=(I-W')_{\nu}x=b'_{\nu}$ and $b_{[n]\setminus\nu}=0$. Then by construction
\[x=(I_{\nu}-W_{\nu})^{-1}b_{\nu}>0, \textnormal{ and } 0\le b_{[n]\setminus\nu}\le -W_{([n]\setminus\nu)\nu}x,\]
and thus $\nu=\sigma\sqcup\tau'\in FP(W,b)\subset FP(W)$ by \cref{lemma:fixed_point_conditions}. Therefore $\sigma\in \exccode{W}$. 

In the case that $(n+1)\in \tau'$, let \new{$\mathcal{E_{\tau'}}$} be the set of indices of nonzero rows of $W'_{\exc\tau'}$ \new{(the set of inhibited excitatory neurons of the Dale network $W'_{\exc\tau'}$)}. For each $i\in \alpha$ let $j_i\in \inhib$ be such that $W'_{ij_i}=W_{ij_i}<0$ (We can always choose $j_i$ to not equal $n+1$ by the hypothesis of the theorem, and thus we can assume $j_i\in \inhib$). Let $\mu=\sigma\sqcup \inhib$ and define
\begin{align*} c\od
\max \Bigg\{\, 
&\max_{i\in \sigma\cap \new{\mathcal{E_{\tau'}}}}-\frac{1}{W_{ij_i}}\Bigl\vert\sum_{k\in \sigma}(I-W)_{ik}x_k\Bigr\vert, 
\, \, \max_{i\in \inhib}\Bigl\vert\sum_{k\in \sigma}-W_{ik}x_k\Bigl\vert,\\
&\max_{i \in([n]\setminus\mu)\cap \alpha}-\frac{1}{W_{ij_i}}\Bigl\vert\sum_{k\in \sigma}-W_{ik}x_k\Bigr\vert\,\Bigg \}.
\end{align*}

Let $y\in \mathbb{R}^{|\mu|}_{>0}$ be defined by $y_{\sigma}=x_{\sigma}$ and $y_j=1+c$ for $j \in \inhib$. \new{We aim to show that $(I-W)_{\mu}y>0$ and that $0\le -W_{([n]\setminus \mu)\mu}y$ as that will give us that $\mu\in FP(W)$ by \cref{lemma:fixed_point_conditions}. We prove these inequalities by considering the partitions 
\[\mu=\sigma\sqcup \inhib=(\sigma\cap \mathcal{E_{\tau'}})\sqcup (\sigma\setminus \mathcal{E_{\tau'}})\sqcup\inhib, [n]\setminus \mu=(([n]\setminus \mu)\cap \mathcal{E_{\tau'}})\sqcup (([n]\setminus \mu)\setminus \mathcal{E_{\tau'}}).\]}

Then for $i\in \sigma\cap \mathcal{E_{\tau'}}$, we have \new{
\begin{align*}
    ((I-W)_{\mu}y)_i=\sum_{k\in \mu}(I-W)_{ik}y_k=\sum_{k\in \sigma}(I-W)_{ik}y_k+\sum_{k\in \inhib}(I-W)_{ik}y_k=\\
    =\sum_{k\in\sigma} (I-W)_{ik}x_k+(1+c)\sum_{k\in\inhib}(I-W)_{ik}=\\
    =\sum_{k\in \sigma}(I-W)_{ik}x_k+(1+c)(I-W)_{ij_i}+(1+c)\sum_{j_i\neq k\in \inhib}(I-W)_{ik}\ge\\
    \ge\sum_{k\in \sigma}(I-W)_{ik}x_k+c(-W_{ij_i})-W_{ij_i}+(1+c)0\ge\\
    \ge\sum_{k\in \sigma}(I-W)_{ik}x_k+(-W_{ij_i})(-\frac{1}{W_{ij_i}})\left\vert\sum_{k\in \sigma} (I-W)_{ik}x_k\right\vert-W_{ij_i}\ge \\
    \ge 0-W_{ij_i}>0.
\end{align*}}
For $i\in \sigma\setminus \new{\mathcal{E_{\tau'}}}$, we have
\begin{align*}
    ((I-W)_{\mu}y)_i=\sum_{k\in \mu}(I-W)_{ik}y_k=\sum_{k\in \sigma}(I-W)_{ik}y_k+\sum_{k\in \inhib}(I-W)_{ik}y_k=\\
    =\sum_{k\in \sigma}(I-W)_{ik}x_k+\sum_{k\in \inhib}(I-W)_{ik}y_k\ge \sum_{k\in \sigma}(I-W)_{ik}x_k+0>0,
\end{align*}
because the $i$-th row of $W'_{\exc\tau'}$ is zero for $i\not \in \new{\mathcal{E_{\tau'}}}$ by assumption and $x$ was by hypothesis such that for $i\in \sigma\setminus\new{\mathcal{E_{\tau'}}}$, $((I-W')_{\nu}x)_i=\sum_{k\in\nu} (I-W')_{ik}x_k=\sum_{k\in \sigma}(I-W')_{ik}x_k+\sum_{k\in \tau'}(I-W')_{ik}x_k=\sum_{k\in \sigma}(I-W)_{ik}x_k+0>0$.

\new{Let $i\in \inhib$. Since $i\not\in \sigma$, we have $\sum_{k\in \sigma}(I-W)_{ik}x_k=\sum_{k\in \sigma}-W_{ik}x_k$. Furthermore, note that $\sum_{i\neq k\in\inhib}-W_{ik}\ge 0$. From all of this we have
\begin{align*}
    ((I-W)_{\mu}y)_i=\sum_{k\in \mu}(I-W)_{ik}y_k=\sum_{k\in \sigma}(I-W)_{ik}y_k+\sum_{k\in \inhib}(I-W)_{ik}y_k=\\
    =\sum_{k\in \sigma}(I-W)_{ik}x_k+(1+c)\sum_{k\in \inhib}(I-W)_{ik}=\\
    =\sum_{k\in \sigma}-W_{ik}x_k+(1+c)((I-W)_{ii}+\sum_{i\neq k\in \inhib}-W_{ik})\ge\\
    \ge \sum_{k\in \sigma}-W_{ik}x_k+1+c
    \ge \sum_{k\in \sigma}-W_{ik}x_k+1+\left\vert\sum_{k\in \sigma}-W_{ik}x_k\right\vert>0.
\end{align*}}
Together, all these inequalities yield that $(I-W)_{\mu}y>0$ and $(I_{\mu}-W_{\mu})^{-1}(I-W)_{\mu}y=y>0$. 

On the other hand, for $i\in ([n]\setminus \mu)\cap \new{\mathcal{E_{\tau'}}}$, we have
\begin{align*}
    (-W_{([n]\setminus\mu)\mu}y)_i
    = (-W_{([n]\setminus\mu)\mu}y)_i
    =\sum_{k\in\mu}-W_{ik}y_k=\\
    =\sum_{k\in\sigma}-W_{ik}y_k+\sum_{k\in \inhib} -W_{ik}y_k=\sum_{k\in\sigma}-W_{ik}x_k+\sum_{k\in \inhib} -W_{ik}+c\sum_{k\in \inhib}-W_{ik}=\\
    =\sum_{k\in\sigma}-W_{ik}x_k-W_{ij_i}+\sum_{j_i\neq k\in \inhib}-W_{ik}+c(-W_{ij_i})+c\sum_{j_i\neq k\in \inhib}-W_{ik}\ge\\
    \ge\sum_{k\in \sigma}-W_{ik}x_k-W_{ij_i}+0+(-W_{ij_i})(-\frac{1}{W_{ij_i}})\left\vert\sum_{k\in \sigma}-W_{ik}x_k\right\vert+0\ge\\
    \ge -W_{ij_i}>0.
\end{align*}
Furthermore, for $i\in ([n]\setminus\mu)\setminus \new{\mathcal{E_{\tau'}}}$, we have
\begin{align*}
    (-W_{([n]\setminus\mu)\mu}y)_i
    =\sum_{k\in\mu}-W_{ik}y_k=\\
    =\sum_{k\in\sigma}-W_{ik}y_k+\sum_{k\in \inhib} -W_{ik}y_k\ge \sum_{k\in\sigma}-W_{ik}x_k+0\ge 0,
\end{align*}
because the $i$-th row of $W'_{\exc\tau'}$ is zero for $i\not \in \new{\mathcal{E_{\tau'}}}$ by assumption and $x$ was by hypothesis such that for $i\in ([n]\setminus\mu)\setminus\new{\mathcal{E_{\tau'}}}\subset \exc$,
\[(-W'_{([n+1]\setminus\nu)\nu}x)_i=\sum_{k\in \nu}-W'_{ik}x_k=\sum_{k\in \sigma}-W'_{ik}x_k+\sum_{k\in \tau'}-W'_{ik}x_k=\sum_{k\in \sigma}-W_{ik}x_k+0\ge 0.\]
Let $b\in \mathbb{R}^{n}_{\ge 0}$ be defined by $b_{\mu}=(I-W)_{\mu}y$ and $b_{[n]\setminus\mu}=0$. Then, by construction $\mu=\sigma\sqcup \inhib\in FP(W,b)\subset FP(W)$ by \cref{lemma:fixed_point_conditions}. Therefore $\sigma\in \exccode{W}$.
\end{proof} 

\begin{corollary}
\label{corollary:replacing_all_inhibitory_neurons_by_one}
Let $W\in \mathbb{D}_{n}$, and $W'\in \mathbb{D}_{n+1}$ be Dale matrices such that $W'_{[n]}=W$, and the last $(n+1)$-st column of $W'$ is inhibitory and such that for all $i\in \exc$, $W_{i(n+1)}<0$ if and only if the $i$-th row vector of $W_{\exc\new{\inhib}}$ is nonzero. Let $W''$ be a Dale matrix obtained by deleting all the inhibitory columns and rows of $W'$ except the $(n+1)$-st row and column. Then $\exccode{W}=\exccode{W''}$.
\end{corollary}

\begin{proof}
By \cref{theorem:code_is_invariant_if_adding_inhib_neuron_with_no_extra_connections} we have that $\exccode{W}=\exccode{W'}$. Now we can reindex the inhibitory columns and rows so that the $(n+1)$-st column of $W'$ becomes the first inhibitory column. Then we start deleting all the other inhibitory rows and columns one at a time and at the end we obtain the matrix $W''$. The code will be preserved at each deletion step by \cref{theorem:code_is_invariant_if_adding_inhib_neuron_with_no_extra_connections} and thus $\exccode{W}=\exccode{W''}$. Note that $W'$ can be generically chosen such that at each deletion step all relevant matrices satisfy the Ground Assumption, which was necessary in the proof of \cref{theorem:code_is_invariant_if_adding_inhib_neuron_with_no_extra_connections}.
\end{proof}

\begin{proof}[Proof of \cref{theorem:replacing_all_inhibitory_neurons_by_one}]
Note that we can obtain $W'$ from $W$ by following the procedure from \cref{corollary:replacing_all_inhibitory_neurons_by_one} where we obtained $W''$. The result thus follows.
\end{proof}

Before proceeding to prove \cref{theorem:replace_inhibitory_rows_by_0}, we make a few observations in \cref{prop:positive_code_upward_closed,corollary:characterization_of_positive_code}
which will simplify our proof strategy.

\begin{proposition}
\label{prop:positive_code_upward_closed}
Let $W\in \Dale{n}$, $\sigma\subset \exc$ and $\tau\subset \inhib$. If $\sigma\sqcup \tau\in FP(W)$, then $\sigma\sqcup \tau'\in FP(W)$ for any $\tau\subset \tau'\subset \inhib$.
\end{proposition}

\begin{proof}
Suppose that $\nu=\sigma\sqcup \tau\in FP(W)$ and let $\tau\subset\tau'\subset \inhib$, $\nu'=\sigma\sqcup \tau'$.
Then, there exists a $b\in \mathbb{R}^n_{\ge 0}$ such that \[x\od(I_{\nu}-W_{\nu})^{-1}b_{\nu}>0, \textnormal{ and } 0\le b_{\overline{\nu}}\le -W_{\overline{\nu}\nu}x,\]
by \cref{lemma:fixed_point_conditions}.
Define
\[c\od
\max_{i\in \tau'\setminus \tau}\left\vert\sum_{k\in \sigma}(I-W)_{ik}x_k\right\vert .\]
Let $x'\in \mathbb{R}^{|\nu'|}_{>0}$ be defined by $x'_{\nu}=x$ and $x'_i=1+c$ for $i\in \nu'\setminus\nu$. Then for $i\in \nu$ we have 
\begin{align*}
    ((I-W)_{\nu'}x')_i=\sum_{k\in \nu'}(I-W)_{ik}x'_k=\\
    =\sum_{k\in \nu}(I-W)_{ik}x'_k+\sum_{k\in \nu'\setminus\nu}(I-W)_{ik}x'_k\ge\\
    \ge\sum_{k\in \nu}(I-W)_{ik}x_k+0=b_i\ge 0.
\end{align*}
For $i\in \nu'\setminus\nu$ we have 
\begin{align*}
    ((I-W)_{\nu'}x')_i=\sum_{k\in \nu'}(I-W)_{ik}x'_k=\\
    =\sum_{k\in \sigma}(I-W)_{ik}x'_k+\sum_{k\in \tau}(I-W)_{ik}x'_k+\sum_{k\in \tau'\setminus\tau}(I-W)_{ik}x'_k\ge \\
    =\sum_{k\in \sigma}(I-W)_{ik}x'_k+0+(I-W)_{ii}x'_i+\sum_{i\neq k\in \tau'\setminus\tau}(I-W)_{ik}x'_k\ge\\
    \ge \sum_{k\in \sigma}(I-W)_{ik}x_k +1+c+0
    \ge \sum_{k\in \sigma}(I-W)_{ik}x_k+1+\left\vert\sum_{k\in \sigma}(I-W)_{ik}x_k\right\vert\ge 1>0.
\end{align*}
Furthermore, for $i\in \overline{\nu'}\subset \overline{\nu}$ we have 
\[
    (-W_{\overline{\nu'}\nu'}x')_i=\sum_{k\in \nu'}-W_{ik}x'_k
    =\sum_{k\in \nu}-W_{ik}x'_k+\sum_{k\in \tau'\setminus\tau}-W_{ik}x'_k
    \ge\sum_{k\in \nu}-W_{ik}x_k+0\ge b_i\ge 0.
\]
Define $b'\in \mathbb{R}^{n}_{\ge 0}$ by $b'_{\nu'}=(I-W)_{\nu'}x'$ and $b'_{\overline{\nu'}}=0$. Then by construction and by \cref{lemma:fixed_point_conditions}, $\nu'=\sigma\sqcup \tau'\in FP(W,b')\subset FP(W)$.
\end{proof}


\begin{corollary}
\label{corollary:characterization_of_positive_code}
Let $W\in \mathbb{D}_n$. Then $\sigma\in \exccode{W}$ if and only if $\sigma\sqcup \inhib\in FP(W)$.
\end{corollary}

We now prove \cref{theorem:replace_inhibitory_rows_by_0}.

\begin{theorem}[\cref{theorem:replace_inhibitory_rows_by_0}]
\new{    Let $W\in \mathbb{D}_{n}$ be a Dale matrix with a set of excitatory neurons $\exc$ and inhibitory neurons $\inhib$, and 
 let $W'\in \mathbb{D}_{n}$ be such that 
\[W'_{ij}=\begin{cases}
0, &\forall i\in \inhib,\, \forall j\in \exc\\
W_{ij}, & \emph{otherwise}
\end{cases}.\]
Then $\exccode{W}=\exccode{W'}$.}
\end{theorem}

\begin{proof}[Proof of \cref{theorem:replace_inhibitory_rows_by_0}]
The proof relies on $W$ satisfying the Ground Assumption, however we will also need $W'$ to satisfy the Ground Assumption. Because of the way $W'$ was constructed from $W$, this is the case. Indeed, observe that for $\sigma\subset [n]$, we have a block matrix decomposition $(I-W')_{\sigma}=\begin{pmatrix}
    (I-W)_{\exc\cap\sigma} & -W_{(\exc\cap\sigma)(\inhib\cap\sigma)}\\
    0 & (I-W)_{\inhib\cap\sigma}
\end{pmatrix}$. Therefore $\det((I-W')_{\sigma})=\det((I-W)_{\exc\cap \sigma})\det((I-W)_{\inhib\cap\sigma})\neq 0$,  by the Ground Assumption for $W$.
Having this in mind, we proceed with the proof.
 
Let $\sigma\in \exccode{W}$. Then there exists a $\tau\subset \inhib$ such that $\nu=\sigma\sqcup \tau\in FP(W)$. Thus, there exists a $b\in \mathbb{R}^n_{\ge 0}$ such that 
\[x\od(I_{\nu}-W_{\nu})^{-1}b_{\nu}>0, \textnormal{ and } -W_{\overline{\nu}\nu}x\ge b_{\overline{\nu}}\ge 0,\]
by \cref{lemma:fixed_point_conditions}.
Observe that because of the way $W'$ was defined we have that 
\[(I-W')_{ij}=\begin{cases}
0, & i\in \inhib,j\in \exc\\
(I-W)_{ij}, & \textnormal{otherwise}
\end{cases}.\] 
Therefore $y\od(I-W')_{\nu}x\ge (I-W)_{\nu}x=b_{\nu}\ge 0$. Observe that $x=(I_{\nu}-W'_{\nu})^{-1}y=(I_{\nu}-W_{\nu})^{-1}b_{\nu}$. Furthermore, $-W'_{\overline{\nu}\nu}x\ge -W_{\overline{\nu}\nu}x\ge 0$. Define $b'\in \mathbb{R}^n_{\ge 0}$ by $b'_{\nu}=y$ and $b'_{\overline{\nu}}=0$. Then by construction $\nu=\sigma\sqcup \tau\in FP(W',b')\subset FP(W')$ by \cref{lemma:fixed_point_conditions}. Therefore $\sigma\in \exccode{W'}$.

Now suppose that $\sigma\in \exccode{W'}$. By \cref{corollary:characterization_of_positive_code} it follows that $\nu=\sigma\sqcup \inhib\in FP(W')$. Thus, there exists a $b'\in \mathbb{R}^n_{\ge 0}$ such that 
\[x'\od(I_{\nu}-W'_{\nu})^{-1}b'_{\nu}>0, \textnormal{ and } -W'_{\overline{\nu}\nu}x'\ge b'_{\overline{\nu}}\ge 0,\]
by \cref{lemma:fixed_point_conditions}. Define
\[c\od
\max_{i\in \inhib}\left\vert\sum_{k\in \sigma}(I-W)_{ik}x'_k\right\vert .\]
Let $x\in \mathbb{R}^{|\nu|}_{>0}$ be defined by $x_{\sigma}=x'_{\sigma}$ and $x_i=x'_i+1+c$ for all $i\in \inhib$. Then, for $i\in \sigma$ we have
\begin{align*}
    ((I-W)_{\nu}x)_i&=\sum_{k\in \nu}(I-W)_{ik}x_k=
    \\&=\sum_{k\in \nu}(I-W')_{ik}x'_k+(1+c)\sum_{k\in \inhib}(I-W)_{ik}\ge
    \\&\ge\sum_{k\in \nu}(I-W')_{ik}x'_k+0\ge 0.
\end{align*}
For $i\in \inhib$ we have 
\begin{align*}
    ((I-W)_{\nu}x)_i=\sum_{k\in \nu}(I-W)_{ik}x_k=
    \sum_{k\in \sigma}(I-W)_{ik}x_k+\sum_{k\in \inhib}(I-W)_{ik}x_k=\\
    =\sum_{k\in \sigma}(I-W)_{ik}x'_k+(I-W)_{ii}x_i+\sum_{i\neq k\in \inhib}(I-W)_{ik}x_k\ge\\
    \ge \sum_{k\in \sigma} (I-W)_{ik}x'_k+(x'_i+1+c)+0\ge \\
    \ge \sum_{k\in \sigma}(I-W)_{ik}x'_k+1+ \left\vert\sum_{k\in \sigma}(I-W)_{ik}x'_k\right\vert\ge 1>0.
\end{align*}
Furthermore, for $i\in \overline{\nu}=\exc\setminus\sigma$ we have
\[
    (-W_{\overline{\nu}\nu}x)_i=\sum_{k\in \nu}-W_{ik}x_k
    =\sum_{k\in \nu}-W'_{ik}x'_k+(1+c)\sum_{k\in \inhib}-W'_{ik}
    \ge \sum_{k\in \nu}-W'_{ik}x'_k+0\ge 0.
\]
Define $b\in \mathbb{R}^{n}_{\ge 0}$ by $b_{\nu}=(I-W)_{\nu}x$ and $b_{\overline{\nu}}=0$. Then by construction and by \cref{lemma:fixed_point_conditions}, $\nu=\sigma\sqcup  \inhib\in FP(W,b)\subset FP(W)$. Therefore, $\sigma\in \exccode{W}$.
\end{proof}

From \cref{theorem:replace_inhibitory_rows_by_0,theorem:replacing_all_inhibitory_neurons_by_one} we have the following corollary.

\begin{corollary}
Let $W\in \mathbb{D}_{n}$, with $\exc=\{1,2,\dots ,m\}$, $|\inhib|=k\ge 1$, $m+k=n$. Let $W'\in \mathbb{D}_{m+1}$ be a Dale matrix obtained from $W$ in the following way: $W'_{ij}=W_{ij}$ for $i,j\in \exc$, and for all $i\in \exc$, $W'_{i(m+1)}<0$ if and only if the $i$-th row of $W_{\exc\inhib}$ is a nonzero row vector, and the $(m+1)$-st row of $W'$ is zero. Then $\exccode{W}=\exccode{W'}$.
\end{corollary}

Before we proceed to prove \cref{theorem:positive_code_characterization}, we will need the following lemma.
\begin{lemma}
\label{lemma:rectangular_semipositivity}
Let $A$ be an $n\times m$ matrix, $n\le m$, such that $\exists \sigma\subset [m]$ with $|\sigma|=n$ and that $A_{[n]\sigma}$ is an invertible $Z$-matrix, and $A_{[n]([m]\setminus\sigma)}\le 0$. Here $A_{[n]\tau}$ is a submatrix of $A$ obtained by deleting columns outside of $\tau$.
Then $\exists x\in \mathbb{R}^m_{>0}$ such that $Ax\ge 0$ if and only if $A_{[n]\sigma}$ is an $M$-matrix.
\end{lemma}

\begin{proof}
We can assume without loss of generality that $\sigma=[n]$, otherwise we can \new{reindex}. 
Suppose that $\exists x\in \mathbb{R}^m_{>0}$ such that $Ax\ge 0$. For all $i \in [n]$ we then have $(Ax)_i=\sum_{k\in [m]}A_{ik}x_k=\sum_{k\in [n]}A_{ik}x_k+\sum_{k\in [m]\setminus[n]}A_{ik}x_k\ge 0$, from where it follows that $\sum_{k\in [n]}A_{ik}x_k=A_{i[n]}x_{[n]}\ge 0$ because $A_{i([m]\setminus[n])}\le 0$ by assumption. Thus, $A_{[n]}x_{[n]}\ge 0$ and thus $A_{[n]}$ is semipositive by \cref{lemma:invertible_semipositive_characterization}. Since it is also a $Z$-matrix \new{(off-diagonal entries are nonpositive)} by assumption, we have that it is an $M$-matrix by \cref{theorem:M_matrix_equivalence}\new{, meaning $\rho(A)<1$.} 

Conversely, suppose that $A_{[n]}$ is an $M$-matrix. Then, since it is also invertible by assumption, from \cref{theorem:M_matrix_equivalence} we have that $A_{[n]}$ is semipositive. Therefore, $\exists y\in \mathbb{R}^n_{>0}$ such that $A_{[n]}y\ge 0$ by \cref{lemma:invertible_semipositive_characterization}. Note that for all $i\in [n]$, $\sum_{k\in [m]\setminus[n]}A_{ik}\le 0$.  Define 
\[M\od\min_{i\in [n]} \sum_{k\in [n]}A_{ik}y_k,\, N\od\min_{i\in [n]}\sum_{k\in [m]\setminus[n]}A_{ik}.\]
 Observe that since $A_{[n]}$ is invertible by assumption and $y>0$ it cannot be that $A_{[n]}y=0$. Thus, $M>0$. Suppose that $N=0$. Since, $A_{[n]([m]\setminus[n])}\le 0$ by assumption, this implies that $A_{[n]([m]\setminus[n])}=0$. Therefore, we can define $x\in \mathbb{R}^m_{>0}$ by $x_{[n]}=y$ and $x_j=1$ for all $j\in [m]\setminus[n]$. Then, by construction $Ax\ge 0$.
Now suppose that $N\neq 0$, that is $N<0$ and define $x\in \mathbb{R}^m_{>0}$ by $x_{[n]}=y$ and $x_j=\frac{-M}{N}$ for all $j\in [m]\setminus[n]$. Then for all $i\in [n]$ we have 
\begin{align*}
    (Ax)_i=\sum_{k\in [m]}A_{ik}x_k=\sum_{k\in [n]}A_{ik}x_k+\sum_{k\in [m]\setminus[n]}A_{ik}x_k=\\
    =\sum_{k\in [n]}A_{ik}y_k-\frac{M}{N}\sum_{k\in [m]\setminus[n]}A_{ik}\ge M-\frac{M}{N}N=M-M=0.
\end{align*}
Therefore $Ax\ge 0$.
\end{proof}

\begin{theorem}[\cref{theorem:positive_code_characterization}]
\new{Let $W $ be a Dale matrix, and $\sigma\subset \exc$ be a non-empty subset of excitatory neurons.   
Then  $  \sigma\in 
\exccode{W}$ if and only if the following two conditions are both satisfied:
\begin{enumerate}
   \item[(i)]  (the spectral condition) $\rho(W_{\new{\excuninhib\cap\sigma}})<1$, 
\item[(ii)] (the graph condition) $\sigma\in \code{G_{\exc}}{\new{\excuninhib}}$,
\end{enumerate}
where $W_{\new{\excuninhib\cap\sigma}}$ denotes the synaptic weights of the excitatory sub-network on the subset $\new{\excuninhib\cap\sigma} $, and  $\rho(W_{\new{\excuninhib\cap\sigma}})$ denotes the spectral radius of the matrix $W_{\new{\excuninhib\cap\sigma}}$.}
\end{theorem}

\begin{proof}[Proof of \cref{theorem:positive_code_characterization}]
By \cref{theorem:replace_inhibitory_rows_by_0,theorem:replacing_all_inhibitory_neurons_by_one}, we can assume that there is only one inhibitory neuron, let us say the $n$-th neuron, and that the $n$-th row of $W$ is all $0$. \new{Recall that $\mathcal{E_U}$ (resp. $\mathcal{E_I}$) denotes the set of uninhibited (resp. inhibited) excitatory neurons of $W$, i.e. $\exc=\mathcal{E_U}\sqcup \mathcal{E_I}$, and that $N_{G_{\exc}}^+(\sigma)$ is the out-neighborhood or the synaptic targets of $\sigma\subset \exc$ in the excitatory connectivity graph $G_{\exc}$ (\cref{def:matrix_graph}).}

Suppose that $\sigma\in \exccode{W}$. 
As observed in \cref{corollary:characterization_of_positive_code}, $\sigma\in \exccode{W}$ if and only if $\tau=\sigma\sqcup \{n\}\in FP(W)$. By \cref{lemma:fixed_point_conditions}, $\tau\in FP(W)$ if and only if $\exists x\in \mathbb{R}^{|\tau|}_{>0}$ such that the following two conditions are satisfied:
\begin{subequations}\begin{align}
\label{eq:fp_cond_1}     (I-W)_{\tau}x\ge 0\\
 \label{eq:fp_cond_2}     0\le -W_{\overline{\tau}\tau}x.
\end{align}
\end{subequations}
For any vector $y\in \mathbb{R}^{|\tau|}_{>0}$, and any $i\in \new{\excuninhib\cap\sigma}$ we have \[((I-W)_{\tau}y)_{i}=\sum_{k\in \tau}(I-W)_{ik}y_k=\sum_{k\in \sigma}(I-W)_{ik}y_k,\]
since $(I-W)_{in}=0$. 
\new{Therefore we can make the following two observations:}

\medskip 
\begin{itemize}[label=$\bullet$,left=0pt]
   \item[(a)]  \new{The condition \eqref{eq:fp_cond_1} ensures that  $x_{\sigma}>0$ satisfies  $(I-W)_{\new{(\excuninhib\cap\sigma)}\sigma}x_{\sigma}\ge 0$. Observe that $\sigma\setminus \excuninhib=\excinhib\cap \sigma$, i.e. $\sigma=(\excuninhib\cap \sigma)\sqcup (\excinhib\cap \sigma)$. Because $(I-W)_{\new{\excuninhib\cap\sigma}}$ is a \new{nonsingular} $Z$-matrix (off-diagonal entries are nonpositive) and $(I-W)_{\new{(\excuninhib\cap\sigma)(\excinhib\cap\sigma)}}\le 0$, by \cref{lemma:rectangular_semipositivity} we have that this is equivalent to $(I-W)_{\new{\excuninhib\cap\sigma}}$ being an $M$-matrix, that is $\rho(W_{\new{\excuninhib\cap\sigma}})<1$. }\\
   
    \item[(b)] \new{By observing that $\overline{\tau}=\exc\setminus\sigma$, the condition \eqref{eq:fp_cond_2} is equivalent that for all $i\in \exc\setminus\sigma$ we need to have $(-W_{(\exc\setminus\sigma)\tau}x)_i=\sum_{k\in \tau}-W_{ik}x_k\ge 0$. For $i\in \exc\setminus\sigma$, if $W_{in}=0$ the only possibility is that $W_{i\sigma}=0$ for the inequality to be true. Thus to be able to find an $x\in \mathbb{R}^{n}_{>0}$ for which condition \eqref{eq:fp_cond_2} is satisfied we need that for all $i\in \exc\setminus\sigma$, $W_{i\sigma}=0$ or $W_{in}<0$. 
This implies that $\new{\outnbhd{G_{\exc}}{\sigma}}\cap\new{\excuninhib}\subset \sigma$.}
\end{itemize}

\medskip 

\noindent Now suppose the following two conditions are both satisfied:
\begin{subequations}\begin{align}
\label{eq:spec_cond}     \rho(W_{\new{\excuninhib\cap\sigma}})<1\\
 \label{eq:graph_cond}     \new{\outnbhd{G_{\exc}}{\sigma}}\cap\new{\excuninhib}\subset\sigma.
\end{align}
\end{subequations}
Condition \eqref{eq:spec_cond} means that $(I-W)_{\new{\excuninhib\cap\sigma}}$ is an $M$-matrix. \new{Since} $(I-W)_{\new{(\excuninhib\cap\sigma)(\excinhib\cap\sigma)}}\le 0$, by \cref{lemma:rectangular_semipositivity} this means that there $\exists x\in \mathbb{R}^{|\sigma|}_{>0}$ such that $(I-W)_{\new{(\excuninhib\cap\sigma)}\sigma}x\ge 0$. 
Let $\tau=\sigma\sqcup \{n\}$, and define 
\begin{align*}
c\od
\max\Bigg\{\, 
&\max_{i\in \new{\excinhib\cap\sigma}}-\frac{1}{W_{in}}\Bigl\vert\sum_{k\in \sigma}(I-W)_{ik}x_k\Bigr\vert, \\
&\max_{i\in \new{\excinhib\setminus\sigma}}-\frac{1}{W_{in}}\Bigl\vert\sum_{k\in \sigma}-W_{ik}x_k\Bigr\vert\, ,\, \, 
\Bigl\vert\sum_{k\in \sigma}-W_{nk}x_k\Bigr\vert\, \,\Bigg\}.
\end{align*}
Define $y\in \mathbb{R}^{|\tau|}_{>0}$ by $y_{\sigma}=x$ and $y_{|\tau|}=1+c$. Then, for $i\in \new{\excuninhib\cap\sigma}$ we have 
\begin{align*}
    ((I-W)_{\tau}y)_i=\sum_{k\in \tau}(I-W)_{ik}y_k=\sum_{k\in \sigma}(I-W)_{ik}y_k+(I-W)_{in}y_{|\tau|}=\\
    =\sum_{k\in \sigma}(I-W)_{ik}x_k-W_{in}(1+c)\ge 0+0(1+c)=0.
\end{align*}
For $i\in \new{\excinhib\cap\sigma}$ we have 
\begin{align*}
    ((I-W)_{\tau}y)_i=\sum_{k\in \tau}(I-W)_{ik}y_k=\sum_{k\in \sigma}(I-W)_{ik}y_k+(I-W)_{in}y_{|\tau|}=\\
    =\sum_{k\in \sigma}(I-W)_{ik}x_k-W_{in}(1+c)=\sum_{k\in \sigma}(I-W)_{ik}x_k-W_{in}-W_{in}c\ge \\
    \ge \sum_{k\in \sigma}(I-W)_{ik}x_k-W_{in}+W_{in}\frac{1}{W_{in}}\left\vert\sum_{k\in \sigma}(I-W)_{ik}x_k \right\vert\ge -W_{in}>0.
\end{align*}
Similarly if $i=n$, we will also have that $((I-W)_{\tau}y)_i\ge 0$ because of the definition of $y_{|\tau|}$\new{.}
Therefore we have that $(I-W)_{\tau}y\ge 0$. Furthermore, $\new{\outnbhd{G_{\exc}}{\sigma}}\cap\new{\excuninhib}\subset\sigma$ implies that for all $i\in \exc\setminus\sigma$, $W_{i\sigma}=0$ or $W_{in}<0$. Note that $\overline{\tau}=\exc\setminus\sigma=\new{(\excuninhib\setminus\sigma)\sqcup(\excinhib\setminus \sigma)}$. \new{Then} for $i\in  \new{\excuninhib\setminus\sigma}$, and thus $W_{in}=0$, we have that it must be that $W_{i\sigma}=0$ and therefore
\begin{align*}
    (-W_{\overline{\tau}\tau}y)_i=\sum_{k\in \tau}-W_{ik}y_k=\sum_{k\in\sigma}-W_{ik}y_k-W_{in}y_{|\tau|}=0-0(1+c)=0\ge 0.
\end{align*}
On the other hand for $i\in \new{\excinhib\setminus\sigma}$ we have that $W_{in}<0$ and thus
\begin{align*}
    (-W_{\overline{\tau}\tau}y)_i=\sum_{k\in \tau}-W_{ik}y_k=\sum_{k\in\sigma}-W_{ik}y_k-W_{in}y_{|\tau|}=\sum_{k\in \sigma}-W_{ik}x_k-W_{in}(1+c)=\\
    \sum_{k\in \sigma}-W_{ik}x_k-W_{in}-W_{in}c\ge \sum_{k\in \sigma}-W_{ik}x_k-W_{in}+W_{in}\frac{1}{W_{in}}\left\vert\sum_{k\in \sigma}-W_{ik}x_k \right\vert\ge \\
    \ge -W_{in}>0.
\end{align*}
Thus, we have that $0\le -W_{\overline{\tau}\tau}y$. \new{Let $b\in\mathbb{R}^n$ be defined by $b_{\tau}=(I-W)_{\tau}y$ and $b_{\overline{\tau}}=0$.  Then $\tau=\sigma\sqcup\{n\}\in FP(W,b)\subset FP(W)$ (\cref{def:set_of_fixed_pt_supports}) by \cref{lemma:fixed_point_conditions}. Hence $\sigma\in \exccode{W}$.}
\end{proof}

We now prove \cref{theorem:excitatory_code_stability}, mainly relying on results stated in \cref{section:appendix_A}.

\begin{proof}[Proof of \cref{theorem:excitatory_code_stability}]
Since $||W||_F<1$, this implies that $\rho(W_{\eta})<1$ for any principal submatrix $W_{\eta}$, $\eta\subset [n]$. Therefore by \cref{theorem:positive_code_characterization}, $\exccode{W}=\code{G_{\exc}}{\new{\excuninhib}}$.
Furthermore, $||W||_F<1$ implies that $W\in \mathcal{L}$ (\cref{def:L_matrix}) by \cref{lemma:inclusion_between_matrix_classes}. Since the Ground Assumption is assumed to hold, by \cref{proposition:stability_of_switching_system}, for all $b\in \mathbb{R}^n_{\ge 0}$ \eqref{eq:network_dynamics} has a unique globally exponentially stable fixed point. In particular, if $\sigma\in \exccode{W}$, by definition $\sigma\sqcup \tau$ is a support of a fixed point of \eqref{eq:network_dynamics} for some $\tau\subset \inhib$ and some input $b\in \mathbb{R}^n_{\ge 0}$. This fixed point is unique and asymptotically (exponentially) stable as argued above. Therefore $\sigma\in \stbcode{W}$. The inclusion $\stbcode{W}\subset \exccode{W}$ is by definition and thus $\exccode{W}=\stbcode{W}$.
\end{proof}

The proof of \cref{proposition:graph_rules} is straightforward and it gives us the important \cref{corollary:positive_code_is_closed_under_intersections}.

\begin{proof}[Proof of \cref{proposition:graph_rules}]
 It's easy to see that $\varnothing, \exc\in \code{G}{\new{\excuninhib}}$. To show that $\code{\new{G_{\exc}}}{\new{\excuninhib}}$ is respected by intersections and unions 
 assume that $\new{\outnbhd{G_{\exc}}{\sigma_i}}\cap\new{\excuninhib} \subset \sigma_i$  for $i=1,2$, $\sigma_i\subset \exc$. Since $\new{\outnbhd{G_{\exc}}{\sigma_1\cap\sigma_2}}\subset \new{\outnbhd{G_{\exc}}{\sigma_1}} \cap \new{\outnbhd{G_{\exc}}{\sigma_2}} $, we obtain that 
 \begin{equation*} \new{\outnbhd{G_{\exc}}{\sigma_1\cap\sigma_2}}\cap \new{\excuninhib} \subset \new{\outnbhd{G_{\exc}}{\sigma_1}} \cap \new{\outnbhd{G_{\exc}}{\sigma_1}}\cap \new{\excuninhib}\subset \sigma_1\cap\sigma_2. 
 \end{equation*}
 Similarly, 
  \begin{equation*} \new{\outnbhd{G_{\exc}}{\sigma_1\cup\sigma_2}}\cap \new{\excuninhib}=(\new{\outnbhd{G_{\exc}}{\sigma_1}} \cup \new{\outnbhd{G_{\exc}}{\sigma_2}})\cap \new{\excuninhib}=\\
  =(\new{\outnbhd{G_{\exc}}{\sigma_1}}\cap \new{\excuninhib})\cup (\new{\outnbhd{G_{\exc}}{\sigma_2}}\cap \new{\excuninhib})\subset \sigma_1\cup \sigma_2.
   \end{equation*}
\end{proof}

\begin{corollary}
\label{corollary:positive_code_is_closed_under_intersections}
Let $W\in \mathbb{D}_{n}$. The code $\exccode{W}$ is closed under intersections.
\end{corollary}

\begin{proof}
Suppose that $\sigma,\tau\in \exccode{W}$. By \cref{theorem:positive_code_characterization} we have that $\rho(W_{\new{\excuninhib\cap\sigma}})$, $\rho(W_{\new{\excuninhib}_{\tau}})<1$ and that $\sigma,\tau\in \code{G_{\exc}}{\new{\excuninhib}}$. Note that $W_{\new{\excuninhib\cap(\sigma\cap \tau)}}$ is a principal submatrix of $W_{\new{\excuninhib\cap\sigma}}$ which is nonnegative. Thus, by \cref{lemma:eigenradius_of_submatrix}, $\rho(W_{\new{\excuninhib\cap(\sigma\cap\tau)}})<1$. Furthermore, $\code{G_{\exc}}{\new{\excuninhib}}$ is closed under intersections by \cref{proposition:graph_rules} and thus $\sigma\cap \tau \in \code{G_{\exc}}{\new{\excuninhib}}$. Thus, by \cref{theorem:positive_code_characterization} it follows that $\sigma\cap \tau \in \exccode{W}$.
\end{proof}
By \cref{corollary:positive_code_is_closed_under_intersections,theorem:convex_codes} we immediately get \cref{theorem:positive_code_is_convex}.
\label{proof:theorem:positive_code_is_convex}
\medskip

We now proceed to prove \cref{theorem:constructing_network}. Given a code $C\subset 2^{\exc}$ that is a sublattice with $\varnothing,\exc\in C$, recall the definition of $c:2^{\exc}\to C$ and the graph $(G_c,\exc)$ from \cref{theorem:constructing_network}. Let $\new{\outnbhd{G_c}{\sigma}}$ denote the targets of $\sigma\subset \exc$ (\new{the out-neighborhood of $\sigma$}) in the graph $(G_{c},\exc)$. By definition, one sees that $c(\varnothing)=\varnothing$, $c(\exc)=\exc$, and $\sigma\subset c(\sigma)$. We immediately make the following observations that will help us prove \cref{theorem:constructing_network}.

\begin{lemma}\label{lemma:closure_distributes_over_unions}
Let $C\subset 2^{\exc}$ be a sublattice with $\varnothing,\exc\in C$. Then the following are true.
\begin{enumerate}
    \item Let $\sigma,\tau\subset \exc$. Then $c(\sigma\cup\tau)=c(\sigma)\cup c(\tau)$,
        \item $C=\{\sigma\subset \exc\,|\, \sigma=c(\sigma)\}$.
\end{enumerate}
\end{lemma}

\begin{proof}
\noindent For part 1, we first show that $c(\sigma)\cup c(\tau)\subset c(\sigma\cup \tau)$. Note that since $\sigma\subset \sigma\cup \tau$ it follows that $c(\sigma)\subset c(\sigma\cup \tau)$. Similarly $c(\tau)\subset c(\sigma\cup \tau)$. Therefore $c(\sigma)\cup c(\tau)\subset c(\sigma\cup \tau)$.
Now we show the other inclusion. Observe that since $C$ is a lattice, it follows that $c(\sigma),c(\tau)$ and $c(\sigma)\cup c(\tau)\in C$. Furthermore, $\sigma\cup \tau\subset c(\sigma)\cup c(\tau)$ and therefore by definition $c(\sigma\cup \tau)\subset c(\sigma)\cup c(\tau)$.
    
To prove part 2, suppose that $\sigma\in C$. By definition, \[c(\sigma)=\bigcap\limits_{\nu\in C,\sigma\subset \nu}\nu.\] 
Since $\sigma\in C$, it follows that 
\[\bigcap\limits_{\nu\in \mathcal{C},\sigma\subset \nu}\nu=\sigma,\] and thus $c(\sigma)=\sigma$.
Now suppose that $\sigma\subset\exc$ is such that \[c(\sigma)=\bigcap\limits_{\nu\in C,\sigma\subset \nu}\nu=\sigma.\] 
Since $C$ is finite and is closed under finite intersections, it follows that $\sigma\in C$.
\end{proof}

\begin{lemma}
\label{lemma:targets_are_transitive}
    The arcs in $G_c$ are a transitive relation on $\exc$.
\end{lemma}

\begin{proof}
    Let $i,j,k\in \exc$ be such that $i\to j$ and $j\to k$. In other words, we have $j\in c(i)$ and $k\in c(j)$. Let $\nu\in C$ be such that $i\in \nu$. Then, by definition $j\in \nu$. Since $k\in c(j)$, by definition this implies that $k\in \nu$. Therefore, by definition $i\to k$.
\end{proof}

\begin{lemma}
\label{lemma:constructing_network}
    Let $G$ be any graph whose transitive closure is $G_c$. Then $\code{G}{\exc}=\code{G_c}{\exc}$
\end{lemma}

\begin{proof}
    Let $\new{\outnbhd{G}{\sigma},\outnbhd{G_c}{\sigma}}$ be the targets of $\sigma\subset \exc$, in $G$ and $G_c$ respectively. By definition $\code{G}{\exc}=\{\sigma\subset\exc\,|\,\new{\outnbhd{G}{\sigma}}\subset\sigma\}$  and $\code{G_c}{\exc}=\{\sigma\subset\exc\,|\,\new{\outnbhd{G_c}{\sigma}}\subset\sigma\}$. By assumption $\new{\outnbhd{G}{\sigma}}\subset \new{\outnbhd{G_c}{\sigma}}$ for all $\sigma\subset \exc$. Thus for any $\sigma\subset\exc$ if $\new{\outnbhd{G_c}{\sigma}}\subset\sigma$ then $\new{\outnbhd{G}{\sigma}}\subset\sigma$. Therefore $\code{G_c}{\exc}\subset \code{G}{\exc}$. Furthermore, if for any $\sigma\subset \exc$, $\new{\outnbhd{G}{\sigma}}\subset \sigma$, then $\new{(N_{G_c}^+)}^m(\sigma)\subset \sigma$ where $\new{(N_{G_c}^+)}^m$ is an $m$-fold application of the $\new{N_{G_c}^+}$ operator to $\sigma$, for any $m\ge 1$. Thus, $\new{N_c^+}(\sigma)\subset \sigma$. Therefore $\code{G}{\exc}\subset \code{G_c}{\exc}$. 
\end{proof}

Finally, we can prove \cref{theorem:constructing_network}.

\begin{proof}[Proof of \cref{theorem:constructing_network}]
     By definition, 
     \[\new{\code{G_c}{\exc}=\{\sigma\subset \exc\,|\, \outnbhd{G_c}{\sigma}\cap \exc=\outnbhd{G_c}{\sigma}\subset \sigma\}.}\] 
     Thus, we need to show that $\{\sigma\subset \exc\,|\, \new{\outnbhd{G_c}{\sigma}}\cap \exc=\new{\outnbhd{G_c}{\sigma}}\subset \sigma\}=\{\sigma\subset\exc\,|\, c(\sigma)=\sigma\}$. Let $G$ be the digraph on $\exc$ be defined by $i\to j$ if and only if $j\in c(i)$ for all $i,j\in \exc$. Let $\new{\outnbhd{G}{\sigma}}$ be the targets of $\sigma\subset \exc$ in $G$. Note that by construction, $\new{\outnbhd{G}{\sigma}=\outnbhd{G_c}{\sigma}}\cup \sigma$. Therefore for all $\sigma\subset \exc$, $\new{\outnbhd{G}{\sigma}}=\sigma$ if and only if $\outnbhd{G_c}{\sigma}\subset \sigma$. Furthermore, by construction $\outnbhd{G}{i}=c(i)$. Hence by \cref{lemma:closure_distributes_over_unions} for all $\sigma\subset \exc$ we have 
    \[\new{\outnbhd{G}{\sigma}}=\bigcup_{i\in \sigma}\new{N_G^+}(i)=\bigcup_{i\in \sigma}c(i)=c\left(\bigcup_{i\in \sigma}i\right)=c(\sigma).\]
    Therefore for all $\sigma\subset \exc$, $c(\sigma)=\sigma$ if and only if $\new{\outnbhd{G_c}{\sigma}}\subset \sigma$.
\end{proof}

\appendix
 
\section{Stability of Linear-Threshold Rate Dynamics}
\label{section:appendix_A} Here we \new{recall} the necessary  results on the stability of the linear threshold dynamics  from \cite{nozari2020hierarchical,pavlov2005convergent} that we used in the proofs of Section \ref{section:proofs}. 
For a matrix $A$ let $||A||$ be its $2$-norm and let $|A|$ denote the matrix $|A|_{ij}=|A_{ij}|$. \new{
Given a $\sigma\in \{0,1\}^n$ let $\Sigma=\textnormal{diag}(\sigma)$, that is the diagonal matrix with the elements of $\sigma$ on the diagonal.}

\begin{definition}
\label{def:L_matrix}
An $n\times n$ matrix $A$ is totally-$\mathcal{L}$ stable, written $A\in \mathcal{L}$, if there exists $P=P^T\succ 0$ such that $(-I+A^T\Sigma)P+P(-I+\Sigma A)\prec 0$ for all $\Sigma=\textnormal{diag}(\sigma) $ and $\sigma\in \{0,1\}^n$.
\end{definition}

\begin{lemma}{\cite[Lemma 2.3]{nozari2020hierarchical}}
\label{lemma:inclusion_between_matrix_classes}
\begin{enumerate}
    \item $\rho(|W|)<1\Longrightarrow W\in \mathcal{L}$.
    \item $||W||<1\longrightarrow W\in\mathcal{L}$.
\end{enumerate}
\end{lemma}

In the \new{statement} of the following \new{result} the authors in \cite{nozari2020hierarchical} had a hypothesis that $W$ is \new{nonsingular} and that for all $\sigma\in \{0,1\}^n$,  $(I-\Sigma W)$ is \new{nonsingular} as well. 
By the standard correspondence between finite binary sequences and finite subsets of $2^{[n]}$, it is not hard to observe that for a given $\sigma\in \{0,1\}^n$, $(I-\Sigma W)$ is \new{nonsingular} if and only if $(I-W)_{\sigma}$ is \new{nonsingular}. Thus, the second hypothesis is equivalent to the Ground Assumption. We thus continue not writing out the Ground Assumption in the statement of the following theorem as is the case in the rest of the paper, but the reader should note that it is indeed necessary for the statement to be true.

\begin{proposition}{\cite[Proposition 4.9]{nozari2020hierarchical}}
    \label{proposition:stability_of_switching_system}
    Consider the network dynamics in \eqref{eq:network_dynamics}. If
$\rho(|W|) < 1$ or $||W|| <1$, then for all $b \in R^{n}$, the network
has a unique fixed point $x^*$ and it is globally exponentially stable relative to $x^*$.
\end{proposition}

Note that   it has been shown in \cite{nozari2020hierarchical}, that  a more general condition \cite[Theorem 4.8]{nozari2020hierarchical} can also guarantee the above result. However, checking if the matrix $W$ satisfies those conditions   is significantly harder in our context, thus we used the statement  above  instead.

\section*{Acknowledgments}
This work was supported by the NSF Next Generation Networks for Neuroscience Program (award 2014217). \new{We thank the anonymous referees for their many helpful comments that improved the manuscript.}

%% file: References.bib
@article{van1996chaos,
  title={Chaos in neuronal networks with balanced excitatory and inhibitory activity},
  author={Van Vreeswijk, Carl and Sompolinsky, Haim},
  journal={Science},
  volume={274},
  number={5293},
  pages={1724--1726},
  year={1996},
  publisher={American Association for the Advancement of Science}
}

@book{hoppensteadt2012weakly,
  title={Weakly connected neural networks},
  author={Hoppensteadt, Frank C and Izhikevich, Eugene M},
  volume={126},
  year={2012},
  publisher={Springer Science \& Business Media}
}

@article{mastrogiuseppe2018linking,
  title={Linking connectivity, dynamics, and computations in low-rank recurrent neural networks},
  author={Mastrogiuseppe, Francesca and Ostojic, Srdjan},
  journal={Neuron},
  volume={99},
  number={3},
  pages={609--623},
  year={2018},
  publisher={Elsevier}
}

@misc{dale1935pharmacology,
  title={Pharmacology and nerve-endings},
  author={Dale, Henry},
  year={1935},
  publisher={SAGE Publications}
}

@book{osborne2013dale,
  title={Dale's Principle and Communication Between Neurones: Based on a Colloquium of the Neurochemical Group of the Biochemical Society, Held at Oxford University, July 1982},
  author={Osborne, Neville N},
  year={2013},
  publisher={Elsevier}
}

@book{kandel:neural,
  author = {Kandel, Eric R. and Schwartz, James H. and Jessell, Thomas M.},
  title = {Principles of Neural Science},
  publisher = {Elsevier},
  year = {2001}   
}

@article{ashwin2005instability,
  title={When instability makes sense},
  author={Ashwin, Peter and Timme, Marc},
  journal={Nature},
  volume={436},
  number={7047},
  pages={36--37},
  year={2005},
  publisher={Nature Publishing Group UK London}
}

@article{huerta2004reproducible,
  title={Reproducible sequence generation in random neural ensembles},
  author={Huerta, Ram{\'o}n and Rabinovich, Mikhail},
  journal={Physical review letters},
  volume={93},
  number={23},
  pages={238104},
  year={2004},
  publisher={APS}
}

@article{rabinovich2011robust,
  title={Robust transient dynamics and brain functions},
  author={Rabinovich, Mikhail I and Varona, Pablo},
  journal={Frontiers in computational neuroscience},
  volume={5},
  pages={24},
  year={2011},
  publisher={Frontiers Research Foundation}
}

@article{LienkaemperShiuWoodstock,
title = {Obstructions to convexity in neural codes},
journal = {Advances in Applied Mathematics},
volume = {85},
pages = {31-59},
year = {2017},
author = {Caitlin Lienkaemper and Anne Shiu and Zev Woodstock}
}

@misc{TournoyDoiron,
author = {Tournoy, Magnus and Doiron, Brent},
year = {2022},
month = {October},
pages = {},
title = {A Step Towards Uncovering The Structure of Multistable Neural Networks},
doi = {10.48550/arXiv.2210.03241}
}

@Article{CIVY13,
author="Curto, C
and  Itskov,  V
and Veliz-Cuba, A
and Youngs, N",
title="The neural ring: an algebraic tool for analyzing the intrinsic structure of neural codes.",
journal="Bulletin of mathematical biology",
year="2013",
volume="75",
number="9",
pages="1571--1611"
}

@article{giusti2014no,
  title={A no-go theorem for one-layer feedforward networks},
  author={Giusti, Chad and Itskov, Vladimir},
  journal={Neural computation},
  volume={26},
  number={11},
  pages={2527--2540},
  year={2014},
  publisher={MIT Press}
}

@Article{hyperplane2020,
author = "Itskov, V and Kunin, A and Rozen, Z",
title = "Hyperplane neural codes and the polar complex",
year = "2020",
journal = "The Abel Symposia proceedings",
volume = "15",
pages = "343-369"
}

@article{curto2017makes,
  title={What makes a neural code convex?},
  author={Curto, Carina and Gross, Elizabeth and Jeffries, Jack and Morrison, Katherine and Omar, Mohamed and Rosen, Zvi and Shiu, Anne and Youngs, Nora},
  journal={SIAM Journal on Applied Algebra and Geometry},
  volume={1},
  number={1},
  pages={222--238},
  year={2017},
  publisher={SIAM}
}

@article{pattern-completion,
author = {Curto, Carina and Morrison, Katherine},
year = {2015},
month = {12},
pages = {2825-2852},
title = {Pattern Completion in Symmetric Threshold-Linear Networks},
volume = {28},
journal = {Neural Computation}
}

@article{net-encoding,
	Author = {C. Curto and A. Degeratu and V. Itskov},
	Journal = {Neural Comput.},
	Pages = {2858-2903},
	Title = {Encoding binary neural codes in networks of threshold-linear neurons},
	Volume = {25},
	Year = {2013}}

@article{HahnSeungSlotine,
	Author = {R. H. Hahnloser and H.S. Seung and J.J. Slotine},
	Journal = {Neural Comput.},
	Number = {3},
	Pages = {621-638},
	Title = {Permitted and forbidden sets in symmetric threshold-linear networks},
	Volume = {15},
	Year = {2003}}

@article{XieHahnSeung,
	Author = {X. Xie and R. H. Hahnloser and H.S. Seung},
	Journal = {Neural Comput.},
	Pages = {2627-2646},
	Title = {Selectively grouping neurons in recurrent networks of lateral inhibition},
	Volume = {14},
	Year = {2002}}

@article{Hahn2000,
	Author = {R. H. Hahnloser and R. Sarpeshkar and M.A. Mahowald and R.J. Douglas and H.S. Seung},
	Journal = {Nature},
	Pages = {947-951},
	Title = {Digital selection and analogue amplification coexist in a cortex-inspired silicon circuit},
	Volume = {405},
	Year = {2000}}

@article{morrison2024diversity,
  title={Diversity of emergent dynamics in competitive threshold-linear networks},
  author={Morrison, Katherine and Degeratu, Anda and Itskov, Vladimir and Curto, Carina},
  journal={SIAM Journal on Applied Dynamical Systems},
  volume={23},
  number={1},
  pages={855--884},
  year={2024},
  publisher={SIAM}
}

@article{Alexandroff:1937,
	author = {P. {Alexandroff}},
	journal = {Mat. Sbornik N.S.},
	number = {44},
	pages = {501-519},
	title = {{D}iskrete {R}\"aume},
	volume = {2},
	year = {1937}
}

@article{eccles1976electrical,
  title={From electrical to chemical transmission in the central nervous system: the closing address of the sir henry dale centennial symposium cambridge, 19 september 1975},
  author={Eccles, John Carew},
  journal={Notes and records of the Royal Society of London},
  volume={30},
  number={2},
  pages={219--230},
  year={1976},
  publisher={The Royal Society London}
}

@article{ahmadian2015properties,
  title={Properties of networks with partially structured and partially random connectivity},
  author={Ahmadian, Yashar and Fumarola, Francesco and Miller, Kenneth D},
  journal={Physical Review E},
  volume={91},
  number={1},
  pages={012820},
  year={2015},
  publisher={APS}
}

@article{rajan2006eigenvalue,
  title={Eigenvalue spectra of random matrices for neural networks},
  author={Rajan, Kanaka and Abbott, Larry F},
  journal={Physical review letters},
  volume={97},
  number={18},
  pages={188104},
  year={2006},
  publisher={APS}
}

@article{ipsen2020consequences,
  title={Consequences of Dale's law on the stability-complexity relationship of random neural networks},
  author={Ipsen, Jesper R and Peterson, Andre DH},
  journal={Physical Review E},
  volume={101},
  number={5},
  pages={052412},
  year={2020},
  publisher={APS}
}

@article{nozari2020hierarchical,
  title={Hierarchical selective recruitment in linear-threshold brain networks—Part I: Single-layer dynamics and selective inhibition},
  author={Nozari, Erfan and Cort{\'e}s, Jorge},
  journal={IEEE Transactions on Automatic Control},
  volume={66},
  number={3},
  pages={949--964},
  year={2020},
  publisher={IEEE}
}

@inproceedings{pavlov2005convergent,
  title={Convergent piecewise affine systems: analysis and design part i: continuous case},
  author={Pavlov, Alexey and van de Wouw, Nathan and Nijmeijer, Henk},
  booktitle={Proceedings of the 44th IEEE Conference on Decision and Control},
  pages={5391--5396},
  year={2005},
  organization={IEEE}
}

@book{berman1994nonnegative,
  title={Nonnegative matrices in the mathematical sciences},
  author={Berman, Abraham and Plemmons, Robert J},
  year={1994},
  publisher={SIAM}
}

@article{curto2012flexible,
  title={Flexible memory networks},
  author={Curto, Carina and Degeratu, Anda and Itskov, Vladimir},
  journal={Bulletin of mathematical biology},
  volume={74},
  number={3},
  pages={590--614},
  year={2012},
  publisher={Springer}
}

@article{cruz2019open,
  title={On open and closed convex codes},
  author={Cruz, Joshua and Giusti, Chad and Itskov, Vladimir and Kronholm, Bill},
  journal={Discrete \& computational geometry},
  volume={61},
  number={2},
  pages={247--270},
  year={2019},
  publisher={Springer}
}

@book{roman2010lattices,
  title={Lattices and Ordered Sets},
  author={Roman, S.},
  year={2010},
  publisher={Springer New York}
}
